\documentclass{amsart}
\usepackage{amssymb,amsmath,amsthm}
\usepackage[cm]{fullpage}
\usepackage{ifthen}
\usepackage{tikz}

\usepackage{booktabs}
\definecolor{regC}{RGB}{46,105,164}   
\definecolor{natC}{RGB}{93,79,157}    
\definecolor{levC}{RGB}{212,120,28}   
\definecolor{incC}{RGB}{176,52,62}    
\definecolor{finC}{RGB}{47,122,84}    
\definecolor{bandN}{RGB}{246,244,250} 
\definecolor{bandR}{RGB}{242,246,251} 
\definecolor{cA}{RGB}{46,105,164}    
\definecolor{cB}{RGB}{212,120,28}    
\definecolor{cC}{RGB}{93,79,157}     
\definecolor{cD}{RGB}{176,52,62}

\usepackage{todonotes}
\usepackage{caption}
\usepackage{subcaption}
\usepackage{floatrow}

\usepackage{xpatch}
\usepackage{graphicx}
\theoremstyle{definition}
\usepackage{comment}
\usepackage{sidecap}
\usepackage{multirow}
\usepackage{enumitem}
\usepackage{xcolor}
\definecolor{myred}{rgb}{0.8,0.1,0.1}

\usepackage{aliascnt}
\usepackage{multirow}
\usepackage{booktabs}

  \definecolor{preC}{RGB}{105,105,115}   
  \definecolor{bandN}{RGB}{246,244,250}
  \definecolor{bandR}{RGB}{242,246,251}

 \usepackage{algorithm}
\usepackage[noend]{algpseudocode}

\algnewcommand\algorithmicinput{\textbf{Input:}}
\algnewcommand\Input{\item[\algorithmicinput]}
\algnewcommand\algorithmicinferred{\textbf{Inferred:}}
\algnewcommand\Inferred{\item[\algorithmicinferred]}
\algnewcommand\algorithmicoutput{\textbf{Output:}}
\algnewcommand\Output{\item[\algorithmicoutput]}

\newcommand{\FRK}[1]{\textcolor{olive}{\textsf{#1}}}
\DeclareMathOperator{\Cov}{Cov}

\usepackage{mathtools}
\DeclarePairedDelimiter{\abs}{\lvert}{\rvert}

\usetikzlibrary{positioning,shapes.geometric,arrows.meta,fit,backgrounds,calc,decorations.pathreplacing}
\usepackage{fancybox}
\usepackage{array}

\usepackage[hidelinks,hypertexnames=false]{hyperref}
\usepackage{cleveref}

\makeatletter
\newcommand{\newsiblingtheorem}[3]{%
  \newaliascnt{#1}{theorem}%
  \newtheorem{#1}[#1]{#2}%
  \aliascntresetthe{#1}%
  \crefname{#1}{#2}{#3}%
  \Crefname{#1}{#2}{#3}%
}
\makeatother

\theoremstyle{plain}
\newtheorem{theorem}{Theorem}[section]
\crefname{theorem}{Theorem}{Theorems}
\Crefname{theorem}{Theorem}{Theorems}

\newsiblingtheorem{proposition}{Proposition}{Propositions}
\newsiblingtheorem{lemma}{Lemma}{Lemmas}
\newsiblingtheorem{corollary}{Corollary}{Corollaries}
\newsiblingtheorem{claim}{Claim}{Claims}
\newsiblingtheorem{conjecture}{Conjecture}{Conjectures}
\newsiblingtheorem{hypothesis}{Hypothesis}{Hypotheses}

\theoremstyle{definition}
\newsiblingtheorem{definition}{Definition}{Definitions}
\newsiblingtheorem{problem}{Problem}{Problems}
\newsiblingtheorem{counterexample}{Counterexample}{Counterexamples}
\newsiblingtheorem{observation}{Observation}{Observations}
\newsiblingtheorem{question}{Question}{Questions}
\newsiblingtheorem{example}{Example}{Examples}
\newsiblingtheorem{input1}{Input1}{Input1s}
\newsiblingtheorem{coupling}{Coupling}{Couplings}

\theoremstyle{remark}
\newsiblingtheorem{remark}{Remark}{Remarks}

\author{Frederik Ravn Klausen*}
\thanks{*frk23@cam.ac.uk, University of Cambridge, DPMMS, Cambridge, United Kingdom}
\address{Frederik Ravn Klausen, University of Cambridge, DPMMS, Cambridge, United Kingdom}
\email{frk23@cam.ac.uk}

\title{Ensuring proportionality: a logical model for compensatory seats added to multi-member constituencies}

\begin{document}
\begin{abstract}
How many compensatory seats are needed to ensure (full) proportionality in two-tier electoral systems with multi-member constituencies?
This question is answered by estimating the discrepancy between seats and votes rigorously, probabilistically and empirically. 

In 1919, P{\'o}lya showed that the Jefferson/D'Hondt method favours large parties. This seat surplus accumulates across constituencies.  In some countries (e.g., Portugal, Poland, Spain and Turkey), this accumulation directly leads to quantifiable disproportionality in the final parliament. Elsewhere compensatory seats counteract this accumulated seat surplus.

Using a logic-first approach, closed formulas for the required number of compensatory seats are derived given regional and national apportionment methods, the number of relevant parties $n$ and the size of the largest party $p$. If there are $c$ constituencies, it is derived logically that a two-tier system built on local D'Hondt apportionment in multi-member constituencies needs around $\frac{c}{2}(n-\frac{1}{p})$ compensatory seats to ensure proportionality. In words, that is half a compensatory seat per constituency times a number that increases with the number of parties and the size of the largest party.

The validity of the derived formulas is analysed empirically using the extensive recent CLEA dataset. Many quantities do not depend much on the total parliament sizes, which means that the variance can be reduced by artificially varying the house sizes. Portugal, Poland, Spain and Turkey all employ a one-tier D'Hondt system and the analysis shows that if they were to add compensatory seats to the system, reasonable numbers would be 40, 75, 75 and 120, assuming that no other properties of the systems change.

The increased mathematical understanding of the number of needed compensatory seats directly informs policy choices that have to be made when adding compensatory seats to a one-tier system.
It also increases the public understanding which is crucial for reforming electoral systems. 
    \end{abstract}
\maketitle

\section{Introduction}
In the Polish election of 2015 the PiS party famously secured a majority (51\% of the seats) in the Sejm with only 37.6\% of the vote, under a system that is called   proportional. How could this  spectacular outcome occur?
Through the bias of the D'Hondt method which gives a seat surplus that accumulates across the 41 Polish constituencies combined with a relatively high electoral threshold. 

To avoid such a mismatch between votes and seats compensatory two-tier systems manage to reconcile proportional and geographical representation by adding compensatory seats \cite{elklit1996category}.
Often, the compensatory seats are enough to fully compensate for the disproportionality arising in the one-tier base. However, in recent years, it has been more common that this compensation could not be carried out fully, in which case the regional and national methods are said to be \emph{incongruent}. This is often a consequence of the increased political fragmentation. 
Incongruence is closely related to, but distinct from \emph{overhang}. Overhang is the name for the additional seats added to parliament that overrepresented parties get in a specific way of resolving the incongruence. Overhang seats are in use in New Zealand and were used in the German Bundestag until 2023, where the system was reformed after the parliament size increased from 598 to 736. 

The recent strain on the mechanisms of two-tier compensatory systems has taken several forms.  
For instance, Sweden reformed the electoral system in 2014 after experiencing incongruence in 2010. In one-tier electoral systems, the increased number of parties does not challenge the mathematical mechanisms, but the disproportionality increases (cf. the Polish Election of 2015). Another instance of this trend is the dramatic 2022 Danish election that was decided by a single incongruent seat \cite{elklit2025hvad}. This event motivated the study \cite{klausensandsynligheden} of the probability of overrepresentation in a  Danish context. 
This paper further develops and generalizes the methods of that paper.

Without a bound on the number of parties it is impossible to bound the number of needed compensatory seats  to avoid incongruence \cite{holdum2025impossibility}. 
Nevertheless, the literature has mentioned $25\%$ of the seats as a rule-of-thumb \cite{aardal2010norske}. 
In this paper, it is argued that one should rather discuss the number of compensatory seats per constituency. With first past the post the percentage is the same as the number of compensatory seats per constituency. This is not true for multi-member constituencies.

The level of disproportionality in one-tier systems depends both on the apportionment method and on the party structure. As it is well known the Jefferson/D'Hondt method has a bias towards the larger parties \cite{polya1919proportionalwahl, schuster2003seat} whereas the Webster/Sainte-Laguë method is unbiased (at least asymptotically \cite{janson2014asymptotic}).  In one-tier systems this disproportionality materialises as \emph{seat surplus} (the number of seats minus the proportional share, see \eqref{eq:seat surplus_definition}). The seat surplus arises as a combination of method bias, malapportionment and fluctuations (for parties above electoral thresholds):
\begin{align}\label{eq:seat_surplus_words}
    \text{seat surplus} = \text{method bias} + \text{malapportionment} + \text{fluctuations}.
\end{align}
Each term can be quantified (cf. \eqref{eq:seat_surplus_split_into_mechanisms}). 
Throughout, $c$ denotes the number of constituencies, $p$ the vote share of the largest party and $n$ is the number of relevant parties (usually this just means the number of parties above the electoral threshold).

This paper focuses on the case of multi-member constituencies and the number of compensatory seats needed to ensure the highest possible degree of proportionality. This is achieved when every party gets at least as many seats when the seats are distributed based on a national tally as when the seats are handed out regionally, we say the two apportionments are \emph{congruent}, see the next sections for further explanations. The requirement on the compensatory tier is essentially determined by the one-tier seat surplus (cf. \Cref{prop:rigorous_bound_on_L}).

Below it will become clear that for the D'Hondt method, the bias is the dominating factor. In that case, it is derived logically (cf. \eqref{eq:expected_number_of_levelling_seats}) that a two-tier system needs around $\frac{c}{2}(n-\frac{1}{p})$ compensatory seats to ensure proportionality.
This prediction is confirmed empirically. 

This law fits into the tradition of logical models of electoral systems \cite{shugart2017votes} and is directly relevant for policy makers and electoral reform advocates.
The law informs estimates on how many compensatory seats are reasonable in multi-member constituencies as they are currently employed in Poland, Spain, Turkey, etc. and informs ongoing debates. 
  
Droop Quota also has a bias towards the largest party and there a similar formula (cf. \eqref{eq:expected_number_of_levelling_seats_quota}) can be derived. 
For the asymptotically unbiased Hare Quota and the Sainte-Laguë methods the fluctuation term usually dominates.


It is also an aim that the understanding that comes from discussing several countries simultaneously will highlight the difference between the countries, which might make it easier for a country which is currently employing the first-past-the-post method to shift to proportional representation. 

\subsection{Omissions}
As electoral systems are extremely diverse we leave out discussions of candidates, two-round systems, non-transferable votes etc. 
The analysis in this paper formally applies to first past the post systems, but it will not be the focus of the paper as the presented bounds are weakest in that case.
Instead we focus on one- and two-tier systems with multi-member constituencies. 

The first-past-the-post case with compensatory seats, known as MMP in the literature, has been particularly important in the analysis of the expanding German parliament and in the introduction of compensatory  seats in New Zealand in 1996. 
In Germany, $60\%$ compensatory seats have been mentioned as an empirical rule to avoid incongruence \cite{bischof2021vierzig}. Furthermore, a logical model has been developed \cite{bochsler2023balancing} to address this case.

\subsection{Outline}
In this paper, we first discuss the mechanisms for disproportionality in one-tier PR electoral systems. Then we give a structural overview of compensatory  two-tier electoral systems. We can then define $L^*$, which is the number of compensatory seats needed to ensure the fullest possible level of proportionality.
It turns out that $L^*$ is very tightly connected to the one-tier seat surplus.
The seat surplus and the bias that creates it is studied both deterministically using the method of Janson \cite{janson2014asymptotic}
 and probabilistically using a generalization of the seat surplus formula from \cite{flis2020pot}. 
These sections give logical predictions for $L^*$ that are then tested empirically in \Cref{sec:empirical_study} using the extensive constituency-level dataset \cite{kollman2024clea}. 
Based on the analysis some conclusions and consequences for policy advice are drawn in \Cref{sec:policy_advice}. There we also put some concrete numbers on the Spanish, Polish and Turkish cases. 
 

\section{Mechanisms of disproportionality of one-tier PR electoral systems} \label{sec:seat surplus_accumulation}
In this section it is discussed how disproportionality can accumulate in one-tier PR electoral systems as a result of method bias, malapportionment, fluctuations and electoral thresholds. 

A one-tier electoral system\footnote{See \cite{holdum2025impossibility} for a rigorous definition.} in $c$ constituencies distributes the seats independently in each constituency using an apportionment method.
The apportionment method may differ from constituency to constituency, but seldom does so.

Depending on the context, the constituency seats $k_1, \dots, k_c$ can be distributed in many ways. For example, in first past the post systems $k_j=1$ for every constituency $j$, while other systems choose to distribute the constituency seats proportionally according to the number of inhabitants or voters in each constituency as will be discussed extensively in \Cref{sec:malapportionment}.

Even though many one-tier electoral systems with multi-member constituencies are usually labelled proportional representation, disproportionality can still arise. The mechanisms for disproportionality of one-tier PR electoral systems are listed below; see  \Cref{fig:one-tier_mechanisms} for an overview. 
\begin{itemize}
    \item[(a)] \emph{Method Bias:} If the one-tier apportionment method has a bias, the resulting seat surplus can accumulate across constituencies.  This is for example the case in Spain and Poland, where the largest parties get significantly overrepresented. 
    If the constituencies are small it can be very difficult for smaller parties to obtain a seat (even for an asymptotically  unbiased method as Sainte-Laguë). This effect, which is sometimes dubbed an effective threshold, can also be understood as a method bias. 
    \item[(b)]  \emph{Malapportionment:} If some regional seats are cheaper than others, parties that perform well in constituencies where seats are discounted can get overrepresented. This malapportionment transfer is strongest if the seats are not distributed to constituencies based on vote totals.  When the seats are cheaper, they are often cheaper in rural areas either due to political negotiation (or due to urbanisation if the seats have not been updated for a long time). In this case, parties that are strong in rural areas can get overrepresented. This was for example the case of the Danish party Venstre until 1970 \cite{klausensandsynligheden}. 
 \item[(c)]  \emph{Fluctuations:} A party could also have won the outermost seats in each of the constituencies and thereby get too many overall (although this is somewhat unlikely - especially when there are many constituencies). 
  \item[(d)]  \emph{Electoral thresholds:}  A one-tier percentage electoral threshold means that only parties with at least a fixed national percentage $\tau$ can win constituency seats\footnote{This should not be confused with a two-tier electoral threshold, where a party can keep regional seats, but is not entitled to any compensatory  seats.}. If $\tau$ is large that might mean that parties can win several percent of the vote without obtaining parliamentary representation at all. One of the best examples is Turkey, which from 1983 employed an electoral threshold of 10\%. As of 2023 the threshold has been lowered to 7\%. Since our analysis below becomes substantially more difficult if there are many small parties that need representation, but cannot win any constituency seats, we will impose an electoral threshold (usually 5\%) in many parts of the analysis. 
\end{itemize}
One aim of the mathematical analysis below is to prove rigorous bounds on the disproportionality of one-tier electoral systems and these mechanisms will emerge naturally from the rigorous framework. 

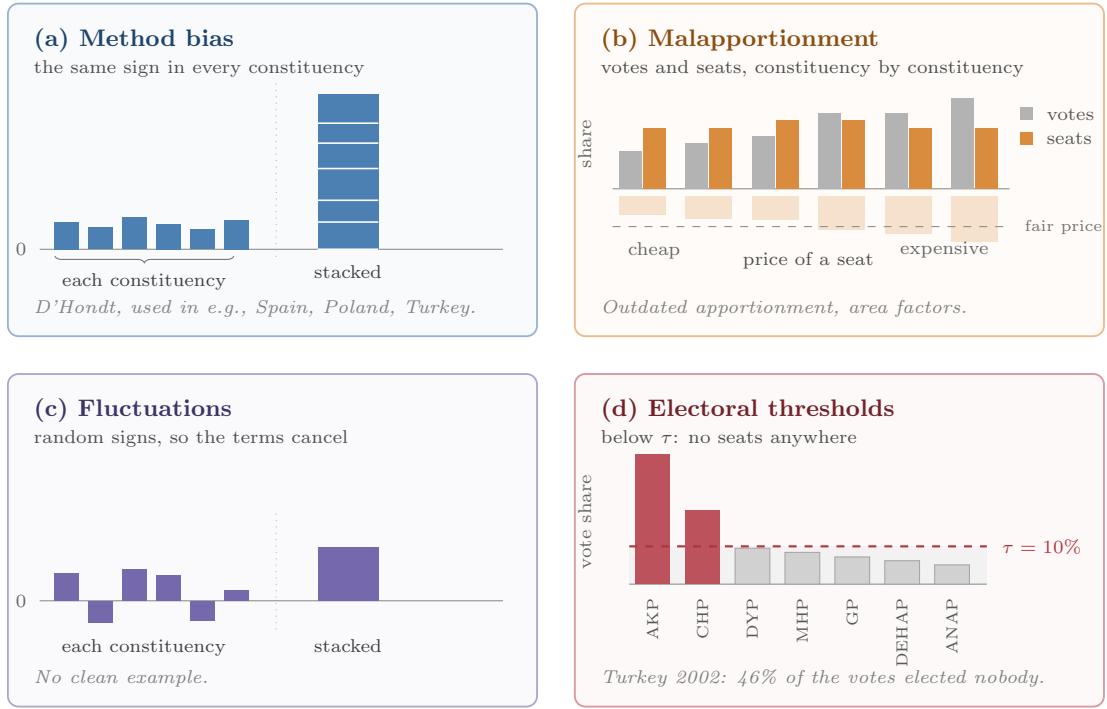
\begin{figure}[ht!]

\begin{tikzpicture}[
  font=\small,
  >={Latex[length=2.0mm]},
  panel/.style={draw=#1!50, fill=#1!3, rounded corners=4pt, line width=0.7pt},
  ptitle/.style={anchor=north west, font=\bfseries\small, text=#1!65!black},
  psub/.style={anchor=north west, font=\scriptsize, text=black!70},
  pex/.style={anchor=south west, font=\scriptsize\itshape, text=black!50},
  tick/.style={font=\scriptsize, text=black!60},
  note/.style={font=\scriptsize, text=black!75, align=center},
]

\begin{scope}[shift={(0,4.9)}]
  \draw[panel=cA] (0,0) rectangle (7.0,4.4);
  \node[ptitle=cA] at (0.22,4.20) {(a) Method bias};
  \node[psub] at (0.22,3.78) {the same sign in every constituency};

  \draw[black!45] (0.42,1.15) -- (6.55,1.15);
  \node[tick,anchor=east] at (0.37,1.15) {$0$};

  \foreach \i/\h in {0/0.30,1/0.24,2/0.35,3/0.28,4/0.22,5/0.32}
    \fill[cA!85] ({0.62+\i*0.45},1.15) rectangle ++(0.32,{\h*1.2});
  \draw[decorate,decoration={brace,mirror,amplitude=2.5pt},black!55]
    (0.62,1.06) -- (2.99,1.06) node[midway,below=3pt,note] {each constituency};

  \draw[black!30,dotted] (3.55,1.05) -- (3.55,3.35);

  \foreach \b/\h in {0.00/0.30,0.30/0.24,0.54/0.35,0.89/0.28,1.17/0.22,1.39/0.32}
    {\fill[cA!85] (4.10,{1.15+\b*1.2}) rectangle ++(0.80,{\h*1.2});
     \draw[white,line width=0.6pt] (4.10,{1.15+\b*1.2}) -- ++(0.80,0);}
  \node[note,anchor=north] at (4.50,1.06) {stacked};

  \node[pex] at (0.22,0.14) {D'Hondt, used in e.g.,\ Spain, Poland, Turkey.};
\end{scope}

\begin{scope}[shift={(7.5,4.9)}]
  \draw[panel=cB] (0,0) rectangle (7.0,4.4);
  \node[ptitle=cB] at (0.22,4.20) {(b) Malapportionment};
  \node[psub] at (0.22,3.78) {votes and seats, constituency by constituency};

  \fill[black!30] (5.90,2.86) rectangle ++(0.16,0.16);
  \node[note,anchor=west,text=black!65] at (6.12,2.94) {votes};
  \fill[cB!85] (5.90,2.54) rectangle ++(0.16,0.16);
  \node[note,anchor=west,text=black!65] at (6.12,2.62) {seats};

  \draw[black!45] (0.50,1.95) -- (5.75,1.95);
  \node[note,rotate=90,anchor=south,text=black!60] at (0.34,2.55) {share};
  \foreach \i/\v/\s in {0/10/16, 1/12/16, 2/14/18, 3/20/18, 4/20/16, 5/24/16}
  {
    \fill[black!30] ({0.58+\i*0.88},1.95) rectangle ++(0.30,{\v*0.05});
    \fill[cB!85]    ({0.90+\i*0.88},1.95) rectangle ++(0.30,{\s*0.05});
  }

  \foreach \i/\p in {0/0.62, 1/0.75, 2/0.78, 3/1.11, 4/1.25, 5/1.50}
    \fill[cB!22] ({0.58+\i*0.88},1.85) rectangle ++(0.62,{-\p*0.40});
  \draw[black!45,dashed] (0.50,1.45) -- (5.75,1.45);
  \node[font=\tiny,text=black!55,anchor=west] at (5.83,1.45) {fair price};

  \node[note,anchor=west,text=black!60] at (0.58,1.14) {cheap};
  \node[note,anchor=east,text=black!60] at (5.60,1.14) {expensive};
  \node[note,anchor=north,text=black!70] at (3.09,1.24) {price of a seat};

  \node[pex] at (0.22,0.14) {Outdated apportionment, area factors.};
\end{scope}

\begin{scope}[shift={(0,0)}]
  \draw[panel=cC] (0,0) rectangle (7.0,4.4);
  \node[ptitle=cC] at (0.22,4.20) {(c) Fluctuations};
  \node[psub] at (0.22,3.78) {random signs, so the terms cancel};

  \draw[black!45] (0.42,1.40) -- (6.55,1.40);
  \node[tick,anchor=east] at (0.37,1.40) {$0$};

  \foreach \i/\h in {0/0.30,1/-0.24,2/0.35,3/0.28,4/-0.22,5/0.12}
    \fill[cC!85] ({0.62+\i*0.45},1.40) rectangle ++(0.32,{\h*1.2});
  \node[note,anchor=north] at (1.80,1.02) {each constituency};

  \draw[black!30,dotted] (3.55,1.02) -- (3.55,2.60);

  \fill[cC!85] (4.10,1.40) rectangle ++(0.80,{0.59*1.2});
  \node[note,anchor=north] at (4.50,1.02) {stacked};

  \node[pex] at (0.22,0.14) {No clean example.};
\end{scope}

\begin{scope}[shift={(7.5,0)}]
  \draw[panel=cD] (0,0) rectangle (7.0,4.4);
  \node[ptitle=cD] at (0.22,4.20) {(d) Electoral thresholds};
  \node[psub] at (0.22,3.78) {below $\tau$: no seats anywhere};

  \fill[black!5] (0.72,1.62) rectangle (5.45,2.12);
  \draw[black!45] (0.72,1.62) -- (5.45,1.62);
  \node[note,rotate=90,anchor=south,text=black!60] at (0.34,2.45) {vote share};

  \foreach \i/\s/\nm/\c in {0/34.3/AKP/in, 1/19.4/CHP/in, 2/9.5/DYP/out,
                            3/8.4/MHP/out, 4/7.2/GP/out, 5/6.2/DEHAP/out,
                            6/5.1/ANAP/out}
  {
    \ifthenelse{\equal{\c}{in}}
      {\fill[cD!85] ({0.80+\i*0.66},1.62) rectangle ++(0.46,{\s*0.05});}
      {\filldraw[fill=black!18,draw=black!35,line width=0.4pt]
         ({0.80+\i*0.66},1.62) rectangle ++(0.46,{\s*0.05});}
    \node[font=\tiny,text=black!65,rotate=90,anchor=east]
      at ({1.03+\i*0.66},1.56) {\nm};
  }

  \draw[cD,dashed,thick] (0.72,2.12) -- (5.45,2.12);
  \node[font=\scriptsize,text=cD,anchor=west] at (5.53,2.12) {$\tau=10\%$};

  \node[pex] at (0.22,0.14) {Turkey 2002: 46\% of the votes elected nobody.};
\end{scope}

\end{tikzpicture}
\caption{Overview of some of the relevant mechanisms of disproportionality for one-tier systems. \label{fig:one-tier_mechanisms}
}
\end{figure}

\section{A structural overview of compensatory  two-tier PR electoral systems}
This section describes the structure of many two-tier compensatory PR systems and introduces the notion of \emph{incongruence} and the number of compensatory seats $L^*$ needed to avoid incongruence.

By definition, compensatory  two-tier electoral systems add compensatory seats to a one-tier system to level out disproportionalities arising from the one-tier basis. Throughout, the number of compensatory seats is denoted by $L$. This letter is chosen because compensatory seats are sometimes called levelling seats (or adjustment seats).

Two-tier electoral systems are often\footnote{Examples that do not fit into this framework are the Polish electoral system for the European parliament and double proportionality.} built on the basis of the following five algorithms.  See also the illustration in \Cref{fig:two_tier_algorithms}, 
\begin{itemize}
    \item[0.] (Constituency algorithm) An apportionment of a total number of constituency seats $k$ to the $c$ constituencies. The algorithm outputs $k_1, \dots, k_c$. An overview of this algorithm will be given in \Cref{tab:seat-apportionment} below.  
    \item[1.] (Regional algorithm) In each constituency this algorithm  determines the distribution constituency seats. Usually, based on an apportionment method applied in every constituency. If there are $k_j$ seats to be distributed in the $j$th constituency the regional algorithm distributes $k_1 + k_2 + \dots + k_c = k$ seats in total. 
    \item[2.] (National algorithm) An algorithm based on the national vote totals. In the cases of interest in this paper, this analysis is with $L$ added compensatory seats, i.e., on a total parliament size of $k+L$.  
    \item[3.] (Incongruence breaker) An algorithm that determines what happens when the regional and national algorithms cannot be reconciled congruently. This algorithm necessarily has to be a bit complicated\footnote{Sometimes these complications and the fact that the algorithms are written in legal texts and not as mathematics or code means that lacunas can arise in the electoral laws, see e.g., \cite{pennisi2006italian,elklit2025hvad,  coppola2026one}.}.  This algorithm is not the object of study in this paper. 
    \item[4.] (National to regional algorithm) Distribution of the national seats to constituencies. Not discussed further in this paper. 
\end{itemize}
The decomposition above focuses on the distribution of seats to parties. Often, there is yet another electoral algorithm that then distributes the party seats to candidates, which is omitted from the analysis of this paper. 


If every party obtains at least as many seats in the national apportionment as in the regional apportionment the outcomes are \emph{congruent}. If not, they are \emph{incongruent} and one or more parties are \emph{overrepresented}. 
Under weak assumptions, for every election outcome, there exists\footnote{This is a consequence of \Cref{prop:rigorous_bound_on_L}. Conversely, the main result of the impossibility theorem for two-tier electoral systems described in \cite{holdum2025impossibility} was to point out that for any electoral system there is no $L$ which is large enough to ensure that incongruence cannot occur for any election outcome.} a large enough number of compensatory seats (the minimal one will be denoted $L^*$) that makes sure that these two algorithms are congruent. 
The goal of this paper is to give rigorous bounds on this number, estimate its expected size and compute it counterfactually for a series of real elections. 

Since we are only interested in congruence the last two algorithms do not matter for us.

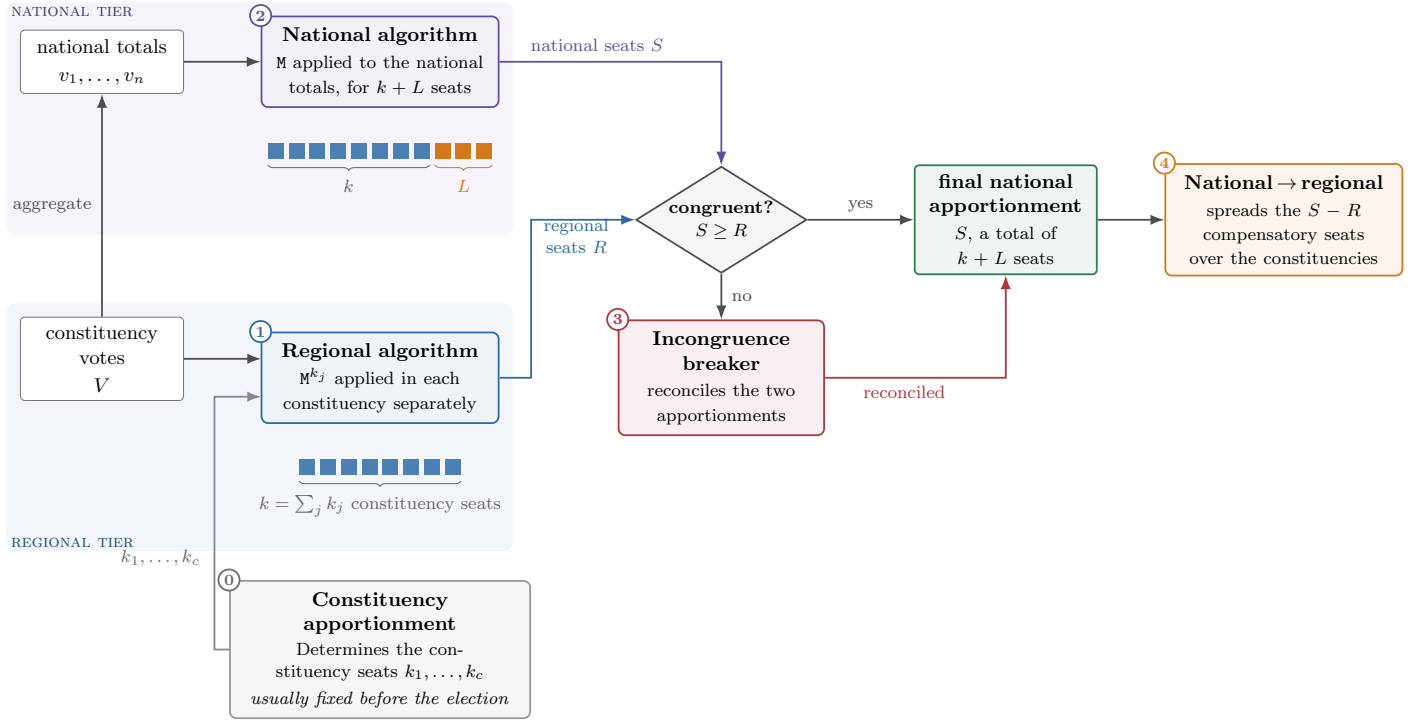
\begin{figure}[ht!]
\centering
\resizebox{\textwidth}{!}{%
\begin{tikzpicture}[
    font=\small,
    >={Latex[length=2.2mm]},
    algo/.style={draw=#1, fill=#1!8, thick, rounded corners=3pt,
                 align=center, inner sep=5pt, text width=34mm},
    pre/.style={draw=preC!70, fill=preC!6, thick, rounded corners=3pt,
                align=center, inner sep=5pt, text width=44mm},
    io/.style={draw=black!55, fill=white, rounded corners=2pt,
               align=center, inner sep=4pt, text width=23mm},
    badge/.style={circle, draw=#1, fill=white, thick, inner sep=1.1pt,
                  font=\scriptsize\bfseries, text=#1},
    lbl/.style={font=\footnotesize, align=center},
    arr/.style={->, thick, black!70},
  ]

\node[algo=regC] (regalgo) at (0,0)
  {\textbf{Regional algorithm}\\[1pt]
   \footnotesize $\mathtt{M}^{k_j}$ applied in each constituency separately};
\node[badge=regC] at (regalgo.north west) {1};

\node[algo=natC] (natalgo) at (0,5.0)
  {\textbf{National algorithm}\\[1pt]
   \footnotesize $\mathtt{M}$ applied to the national totals, for $k+L$ seats};
\node[badge=natC] at (natalgo.north west) {2};

\node[pre] (prealgo) at (0,-4.3)
  {\textbf{Constituency apportionment}\\[1pt]
   \footnotesize Determines the constituency seats $k_1, \dots, k_c$\\[1pt]
   \emph{\footnotesize usually fixed before the election}};
\node[badge=preC] at (prealgo.north west) {0};

\node[io] (const) at (-4.4,0.30) {constituency votes\\[1.5pt] $V$};
\node[io] (nat)   at (-4.4,5.00) {national totals\\[1.5pt] $v_1,\dots,v_n$};

\draw[arr] (const) -- node[lbl, left=1pt] {aggregate} (nat);
\draw[arr] (nat) -- (natalgo);
\draw[arr] (const.east) -- (-1.88,0.30);

\draw[arr, preC!80] (prealgo.west) -- (-2.62,-4.30) -- (-2.62,-0.30) -- (-1.88,-0.30);
\node[lbl, anchor=east, text=preC] at (-2.72,-2.85) {$k_1,\dots,k_c$};

\coordinate (rs) at ($(regalgo.south)+(-12.8mm,-8mm)$);
\foreach \i in {0,...,7}
  \fill[regC!85] ($(rs)+(\i*3.3mm,0)$) rectangle ++(2.5mm,2.5mm);
\draw[decorate,decoration={brace,mirror,amplitude=2.5pt}, black!60]
  ($(rs)+(0,-1mm)$) -- ($(rs)+(7*3.3mm+2.5mm,-1mm)$)
  node[midway, below=3pt, lbl] {$k=\textstyle\sum_j k_j$ constituency seats};

\coordinate (ns) at ($(natalgo.south)+(-17.75mm,-8mm)$);
\foreach \i in {0,...,7}
  \fill[regC!85] ($(ns)+(\i*3.3mm,0)$) rectangle ++(2.5mm,2.5mm);
\foreach \i in {8,9,10}
  \fill[levC] ($(ns)+(\i*3.3mm,0)$) rectangle ++(2.5mm,2.5mm);
\draw[decorate,decoration={brace,mirror,amplitude=2.5pt}, black!60]
  ($(ns)+(0,-1mm)$) -- ($(ns)+(7*3.3mm+2.5mm,-1mm)$)
  node[midway, below=3pt, lbl] {$k$};
\draw[decorate,decoration={brace,mirror,amplitude=2.5pt}, levC]
  ($(ns)+(8*3.3mm,-1mm)$) -- ($(ns)+(10*3.3mm+2.5mm,-1mm)$)
  node[midway, below=3pt, lbl, text=levC] {$L$};

\node[diamond, aspect=1.6, draw=black!70, fill=black!4, thick, align=center,
      inner sep=0.5pt, font=\footnotesize]
  (cons) at (5.4,2.5) {\textbf{congruent?}\\ $S\ge R$};

\draw[arr, regC] (regalgo.east) -- ++(5mm,0) |- (cons.west)
  node[pos=0.45, lbl, right=2pt, text=regC] {regional\\ seats $R$};
\draw[arr, natC] (natalgo.east) -| (cons.north)
  node[pos=0.22, lbl, above=1pt, text=natC] {national seats $S$};

\node[algo=incC, text width=29mm] (inc) at (5.4,0)
  {\textbf{Incongruence breaker}\\[1pt]
   \footnotesize reconciles the two apportionments};
\node[badge=incC] at (inc.north west) {3};
\draw[arr] (cons) -- node[lbl, right=1pt] {no} (inc);

\node[io, draw=finC, fill=finC!8, thick, text width=26mm] (finnat) at (9.9,2.5)
  {\textbf{final national apportionment}\\[1pt]
   \footnotesize $S$, a total of $k+L$ seats};
\draw[arr] (cons) -- node[lbl, above] {yes} (finnat);
\draw[arr, incC] (inc.east) -| node[pos=0.22, lbl, below] {reconciled} (finnat.south);

\node[algo=levC, text width=34mm] (n2r) at (14.3,2.5)
  {\textbf{National\,$\to$\,regional}\\[1pt]
   \footnotesize spreads the $S-R$ compensatory seats over the constituencies};
\node[badge=levC] at (n2r.north west) {4};
\draw[arr] (finnat) -- (n2r);

\begin{scope}[on background layer]
  \node[fit={(nat)(natalgo)($(ns)+(0,-10mm)$)($(ns)+(10*3.3mm+2.5mm,0)$)},
        fill=bandN, rounded corners=4pt, inner sep=6pt] (bandnat) {};
  \node[fit={(const)(regalgo)($(rs)+(0,-10mm)$)($(rs)+(7*3.3mm+2.5mm,0)$)},
        fill=bandR, rounded corners=4pt, inner sep=6pt] (bandreg) {};
  \node[anchor=north west, font=\footnotesize\scshape, text=natC!70!black, inner sep=2pt]
    at (bandnat.north west) {national tier};
  \node[anchor=south west, font=\footnotesize\scshape, text=regC!70!black, inner sep=2pt]
    at (bandreg.south west) {regional tier};
\end{scope}
\end{tikzpicture}
}
\caption{Overview of the mechanics of many two-tier electoral systems. This paper is concerned with how large $L$ has to be for the national and regional apportionment to be congruent. We are not concerned with the algorithms in (3) and (4). See for instance Example 4.1 in \cite{holdum2025impossibility} for a detailed explanation of the Icelandic electoral system, which is arguably the simplest possible two-tier system.    \label{fig:two_tier_algorithms}}
\end{figure}


The regional algorithm (or one-tier system) apportions seats to parties independently in each constituency. For elections for the Spanish \emph{Congreso de los Diputados}, Turkish \emph{Büyük Millet Meclisi}, Polish \emph{Sejm} or Portuguese \emph{Assembleia da República} algorithm 0. and 1. constitute the entire electoral system.
In each of these cases, algorithm 1. is the D'Hondt method. There is no national algorithm and thus the rest of the diagram is irrelevant. 

For compensatory two-tier systems, it is often natural to consider a total of $k$ seats in the constituency-apportionment and compare that to a total of $k+L$ seats in the national apportionment. 
The two apportionments are said to be \emph{congruent} if every party has achieved as many seats in the national apportionment as in the constituency-apportionment.  The additional seats (that one or more overrepresented parties have obtained in the constituency-apportionment) are called \emph{incongruent seats}. 

In most two-tier cases (including Iceland, Denmark, Norway, Sweden, Germany, Austria, New Zealand, etc.), the electoral algorithm terminates if there is no incongruence. 
In that case, the seat distribution of the $k+L$ seats is the final distribution. Each party is assigned compensatory seats corresponding to the difference between the number of national seats and constituency seats. This number is non-negative for all parties if and only if the seat distribution is congruent. 

This paper aims at studying the minimum number of compensatory seats $L$ needed for congruence. For a given election outcome this compensatory requirement is denoted by $L^*$.  In the next section it is defined formally. 

\subsection{Incongruence breakers}
Because our aim is to identify the number of compensatory seats needed for congruence we will only briefly describe possible incongruence-breakers here. 
As discussed in \cite{holdum2025impossibility}, usually it is the case that either the national or the regional level takes priority. 

In the former case, if the national apportionment is with respect to one of the standard methods, the entire electoral system is said to be \emph{guaranteed proportional}. 
In this case, which is employed in Sweden and Germany, some algorithm for redistribution of incongruent regional seats has to be devised. 

If the regional seats take priority and the number of compensatory seats is not sufficient, the incongruence breaker has to prioritize among the compensatory seats. This is most easily done if the national apportionment is by a divisor method, where the quotients give an ordering of the seats. If the national apportionment is using a quota method, iteration is often employed, see \cite{holdum2025impossibility} for additional discussion.

\begin{table}[htbp]
\centering
\small
\setlength{\tabcolsep}{4pt}
\begin{tabular}{@{}l >{\raggedright\arraybackslash}p{0.33\linewidth} l >{\raggedright\arraybackslash}p{0.33\linewidth}@{}}
\toprule
Symbol & Meaning & Symbol & Meaning \\
\midrule 
$c$ & number of constituencies
  & $k$ & total number of constituency (first-tier) seats, $k=\sum_{j} k_j = \sum_i r_i$ \\
$n$ & number of relevant parties (above the threshold $\tau$)
  & $L$ & number of compensatory seats \\
$\tau$ & national electoral threshold (in \%)
  & $L^{*}$ & minimal $L$ avoiding incongruence for a given election outcome \\
$v_i$ & total votes for party $i$, $1\le i\le n$
  & $k+L$ & total parliament size \\
$p_i$ & vote share for party $i$, $1\le i\le n$
  & $\Delta$ & seat surplus \\
$r_i$ & total number of regional seats for party $i$, $1\le i\le n$
  & $\bar p$ & average party constituency vote share, $\bar p = \frac{1}{c} \sum_{j=1}^c p_j$ \\
$k_j$ & number of constituency seats in constituency $j$, $1\le j\le c$
  & $x_j$ & seat discount, $x_j = k_j - \frac{v_j}{v}k$ \\
\bottomrule
\end{tabular}
\caption{Some of the notation used throughout the paper.}
\label{tab:notation}
\end{table}

\section{Definition of \texorpdfstring{$L^*$}{L*}, the number of compensatory seats needed to ensure proportionality.}\label{sec:definition_of_L*}

Apportionment methods are functions from votes to seats in a single constituency. The most important examples are the D'Hondt method (denoted $\mathtt{DH}$), the Sainte-Laguë method ($\mathtt{SL}$), the Hare Quota  largest remainder ($\mathtt{HQLR}$) and the Droop Quota ($\mathtt{DQ}$). 
Whenever an apportionment method, $\mathtt{M}$, apportions $k$ seats we write $\mathtt{M}^k$. For example,  $\mathtt{DH}^5$ is the D'Hondt method which apportions $5$ seats. Given a tuple of votes for $n$ parties $\mathbf{v} =(v_1,v_2, \dots, v_n)$, then $\mathtt{DH}^5(\mathbf{v}) = (s_1, \dots, s_n)$ is the corresponding tuple of seats. Furthermore, $\mathtt{DH}^5(\mathbf{v})_i = s_i$ is the number of seats for the $i$th party.  In this case, the sum of all the $s_i$ equals 5.

\subsection{Definition of \texorpdfstring{$L^*$}{L*}}

Fix an election outcome $\mathbf{v} = (v_1, \dots, v_n)$  consisting of national vote totals and a regional seat allocation $R= (r_1, \dots r_n)$ consisting of a total number of regional seats for each party. Let $k= \sum_i r_i$ be the total number of constituency seats.  For a national apportionment method $\mathtt{M}$. 
Write $\mathtt{M}^{k+L}(\mathbf{v})_i$ for the number of seats that $\mathtt{M}$ apportions to party $i$ when it distributes $k+L$ leasts using the national total.
 $L^*$ is the least number of levelling seats that make the regional and national seat allocations congruent.  That is the smallest $L$ where every party gets at least as many seats in the national allocation as the number of regional seats $r_i$. That is, 
\begin{align}\label{eq:definition_of_L*}
    L_{\mathtt{M}}^*(\mathbf{v},R) \equiv \min \{ L \geq 0 \mid \mathtt{M}^{k+L}(\mathbf{v})_i \geq r_i \text{ for each party } i \}. 
\end{align}
If $L^* \geq 1$ and $\mathtt{M}$ distributes $k+L^*-1$ seats, then at least one party has more regional seats than in the national allocation\footnote{If $\mathtt{M}$ is not house monotone there could be an incongruent $L > L^*$. This will not matter in the analysis below.}. We refer to any such party as the \emph{anchor} party (following the German notion of \emph{Ankerpartei} from \cite{schroder2014parteienproporz}).

We say that $L^*$ is the number of compensatory seats needed to \emph{ensure} proportionality (with respect to $\mathtt{M}$). 
As one would expect, the choice of national method $\mathtt{M}$ matters little in the definition of $L^*$ since the total parliament size ($k$ or $k+L$) is usually relatively large (compared to the average constituency size). 
This is confirmed in \Cref{tab:Lstar_national_methods}.
In contrast, it will turn out that $L^*$ is very sensitive to varying the regional algorithm. 

Throughout, $L^*$ is computed by adding compensatory seats to parliament one by one each time checking whether the regional and national apportionments are congruent. $L^*$ is the first time that happens. This is formalized in \Cref{alg:levelling}, which will be used repeatedly. 

\begin{table}[ht]
\centering
\begin{tabular}{llrrrrrr}
\toprule
Country & Election & $k_0$ & $c$ & $L^*_{\mathrm{DH}}$ & $L^*_{\mathrm{SL}}$ & $L^*_{\mathrm{Hare}}$ & $L^*_{\mathrm{Droop}}$ \\
\midrule
Denmark & 2022 & 135 & 10 & 22 & 25 & 25 & 24 \\
Latvia & 2022 & 100 & 5 & 6 & 8 & 8 & 7 \\
Norway & 2021 & 150 & 19 & 21 & 21 & 21 & 21 \\
Poland & 2023 & 460 & 41 & 57 & 59 & 59 & 58 \\
Portugal & 2024 & 230 & 22 & 7 & 7 & 7 & 7 \\
Spain & 2023 & 350 & 52 & 33 & 34 & 34 & 33 \\
Sweden & 2022 & 310 & 29 & 62 & 63 & 63 & 63 \\
Switzerland & 2023 & 200 & 26 & 20 & 21 & 21 & 20 \\
South Africa & 2019 & 200 & 9 & 4 & 5 & 5 & 4 \\
Turkey & 2023 & 600 & 87 & 85 & 86 & 86 & 85 \\
Peru & 2020 & 130 & 26 & 31 & 29 & 29 & 29 \\
Uruguay & 2019 & 99 & 19 & 20 & 20 & 21 & 20 \\
\bottomrule
\end{tabular}
\caption{Compensatory seats $L^*_{\mathtt{M}}(\mathbf{v}, R)$ with the actual number of constituency seats $k_0$ for the most recent election in the dataset of each country (with a 5\%  threshold imposed). The $k_0$ seats are apportioned across the $c$ constituencies by Sainte-Lagu\"{e} (algorithm 0). While the seats are apportioned to parties within each constituency by D'Hondt. Columns vary the national method $\mathtt{M}$ with respect to which congruence is measured.}
\label{tab:Lstar_national_methods}
\end{table}

\begin{algorithm}[t]
\caption{Minimal number of compensatory seats $L^{*}$}\label{alg:levelling}
\begin{algorithmic}[1]
\Input Total votes for each party $\mathbf{v} = (v_1, \dots, v_n)$.  Total number of one-tier seats for each party $(r_1, \dots, r_n)$. 
National method $\mathtt{M}$. 
\Inferred 
Total number of one-tier seats $k= \sum_{i=1}^n r_i$.  

\Output Minimal $L^{*}$ such that $\mathtt{M}$ applied to the national totals with
       house size $k+L^{*}$ gives every party at least its regional seats.
\State 
Initialize $L=0$
\If {$\mathtt{M}^{k+L}(\mathbf{v})_i \geq r_i$  for each $i= 1, \dots, n$} 
\State \Return $L^{*} \gets L$
\EndIf
\State $L \gets L+1$. 
\State Go to 2. 
\end{algorithmic}
\end{algorithm}

\section{A first analysis of $L^*$}

In this section, $L^*$ is related to proportionality, then rigorous bounds on $L^*$ in terms of the relative seat surplus are given. 
Counter-intuitively, the analysis shows that $L^*$ is smallest when the electoral method in the constituencies has a small bias towards the largest parties. 

\subsection{$L^*$ and proportionality indices}
Earlier work gauged whether the number of compensatory seats suffices for congruence using proportionality indices \cite{shugart1989seats} and \cite{bochsler2023balancing}. 
The index invented by Gallagher\footnote{In our notation, the Gallagher index on the national level is the number, 
$
    \mathtt{GI}(V,S) = \sqrt{\frac{1}{2}\sum_{i=1}^n (p_i - \frac{s_i}{k+L})^2}.  
$
In the plot this number is multiplied by 100.} \cite{gallagher1991proportionality} has now become standard alongside the older Loosemore-Hanby index \cite{loosemore1971theoretical}.

Compensatory seats reduce the disproportionality, but not to 0 since there will still be rounding errors because seats are integers. So where should the cut-off be? The definition of $L^*$ in \eqref{eq:definition_of_L*} circumvents the issue of introducing a cut-off of a proportionality index under which we can say that the electoral system is truly proportional. This is similar to the notion of guaranteed proportionality introduced\footnote{A two-tier system is \emph{guaranteed proportional} if the national totals always correspond to a specific weak-proportional apportionment method (any of the usual methods DH,SL,LR are weakly proportional \cite{balinski2010fair}).
} in \cite{holdum2025impossibility}.   

\Cref{fig:Gallagher_index_vs_levelling_seats} shows how the addition of compensatory seats reduces the Gallagher index up until the point where $L^*$ is reached.
\begin{figure}
     \centering
     \includegraphics[width=\linewidth]{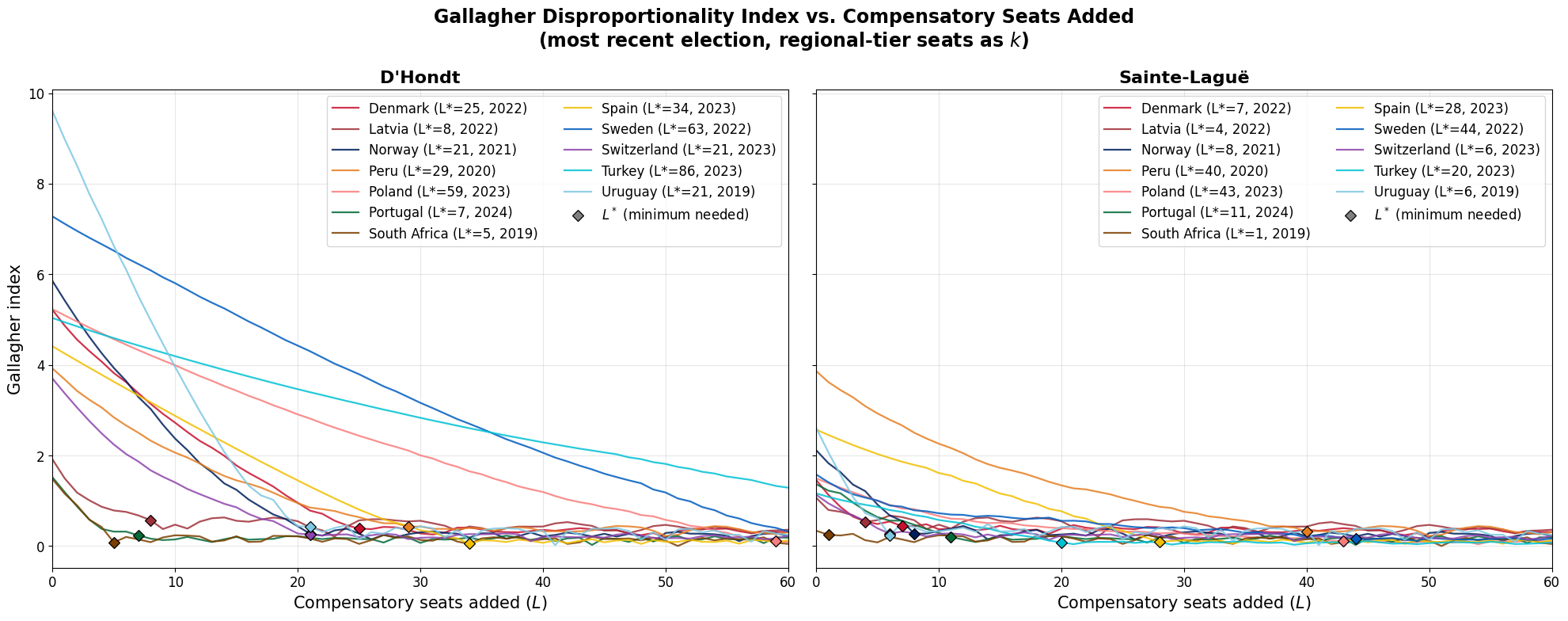}
     \caption{The value of the Gallagher index as a function of the number of compensatory seats added. 
     The index decreases roughly linearly until $L^*$ is reached (marked with diamonds).
     On the left, the D'Hondt method is used to apportion the seats in constituencies and on the right it is the Sainte-Laguë method. Nationally the seats are apportioned using Hare Quota (with iterations as the incongruence breaker, i.e., any party with incongruence is frozen at its regional seat count and new Hare Quota are calculated on the remainders). 
     As we will show later the numbers are generally larger for the D'Hondt method, where the method bias is larger.
\label{fig:Gallagher_index_vs_levelling_seats}}
 \end{figure}
The figure also gives a piece of policy advice: The first compensatory seats are the ones that matter most for the proportionality, since they are the ones most likely to be needed.

\subsection{The seat surplus and its relation to \texorpdfstring{$L^*$}{L*}}
Since the number of compensatory seats needed to ensure proportionality for the D'Hondt method and Droop quota is determined by how much seat surplus is accumulated across constituencies the seat surplus is of central importance. In general, if there are $k$ seats and the party gets a vote share $p$ of the votes its proportional share of the seats is $pk$. 
If the party obtains $s$ seats in the apportionment define the \emph{seat surplus}, $\Delta$ by 
\begin{align}\label{eq:seat surplus_definition}
  \Delta = s - pk.  
\end{align}
As the following proposition shows $L^*$ is more or less determined by the seat surpluses on the one-tier level. This motivates the extensive discussion of this seat surplus in \Cref{sec:seat surplus_accumulation}.

\begin{proposition}\label{prop:rigorous_bound_on_L}
Suppose that party $i$ gets share of the vote $p_i$ and has total seat surplus $\Delta_i$. Then, 
$$\max_{1\leq i \leq n} 
\frac{\Delta_i -1}{p_i} \leq L_{\mathtt{HQLR}}^* \leq  \max_{1\leq i \leq n} \frac{\Delta_i +1}{p_i}.$$ 
\end{proposition}
\begin{proof}
Since Hare Quota is within quota\footnote{See \cite{balinski2010fair} for an extensive discussion of this notion.}  if the $i$th party gets $p_i$ share of the votes then their number of seats in the national allocation $\mathtt{HQLR}^{k+L}(\mathbf{v})_i$ is either $\lfloor p_i(k+L) \rfloor$ or $\lceil p_i(k+L) \rceil$.
In particular, 
$$
 p_i(k+L)-1 \leq \mathtt{HQLR}^{k+L}(\mathbf{v})_i \leq p_i(k+L)+1.
$$
By definition, see \eqref{eq:seat surplus_definition}, the party obtains $r_i = kp_i + \Delta_i $ regional seats.
The condition for congruence is, that for all $i$, 
$\mathtt{HQLR}^{k+L}(\mathbf{v})_i \geq r_i.$ 
This is always satisfied if for all $1\leq i \leq n$, 
$$
 p_i(k+L)-1 \geq kp_i + \Delta_i. 
$$
Thus, $L_{\mathtt{HQLR}}^* \leq \max_{1\leq i \leq n}\frac{\Delta_i +1}{p_i}$. 

Conversely, a sufficient condition for incongruence is that there exists an $i$ such that
$$
 p_i(k+L)+1 < kp_i + \Delta_i.
$$
Thus, for congruence it is necessary that for all $i$ it holds that $\frac{\Delta_i-1}{p_i} \leq L_{\mathtt{HQLR}}^*$. And so $\max_i \frac{\Delta_i-1}{p_i} \leq L_{\mathtt{HQLR}}^*. $
\end{proof}

\Cref{prop:rigorous_bound_on_L} is best illustrated when the apportionment methods can be varied continuously. This is achieved by the $\beta$-divisor methods and $\gamma$-quota methods. 

\subsection{Definition of $\beta$-divisor methods and $\gamma$-quota methods}\label{sec:beta_and_gamma}

Throughout we will consider the framework of $\beta$-divisor methods  and $\gamma$-quota methods, which generalize many of the apportionment methods in use. 
If the reader is familiar with the methods of D'Hondt ($\beta =1$), Sainte-Laguë ($\beta = \frac{1}{2}$), Droop ($\gamma =1$) and Hare ($\gamma =0$), but not interested in the mathematical generalities this subsection can be skipped upon taking note of the parentheses above in the following. 

For the introduction we follow the excellent references  \cite{janson2014asymptotic, pukelsheim2017proportional}.
Combining the notation of \cite{janson2014asymptotic} and \cite{holdum2025impossibility}, we define the $\beta$-divisor method $ \mathtt{Div}^{k}_\beta, $
as the function which apportions $k$ seats to a single-constituency election outcome $(v_1, \dots, v_n)$ such that the number of seats of party $i$ is given by
$$
s_i = \lceil \frac{v_i}{D} - \beta \rceil 
$$
where the divisor $D$ is chosen such that $k$ seats are apportioned, i.e., $\sum_{i=1}^n s_i = k$. 
When $\beta = 1$ this method corresponds to the D'Hondt method $\mathtt{DH}^k$ and when $\beta = \frac{1}{2}$ this corresponds to the Sainte-Laguë method. 
There are many other divisor methods, most importantly Huntington's, but it does not fit into the framework of $\beta$-divisor methods (where the bias can be quantified and which is particularly nice \cite{balinski2014parametric}). 

Similarly, we define the $\gamma$-quota method $\mathtt{Quota}_\gamma^k$  as the function which apportions $k$ seats for the election outcome  $(v_1, \dots, v_n)$  with $V = \sum_{i=1}^n v_i$ through the formula 
$$
s_i = \lceil  (k+ \gamma) \frac{v_i}{V}  - \alpha \rceil 
$$
where $\alpha$  is chosen such that $\sum_{i=1}^n s_i = k$. 

Important special cases include $\gamma =0$, $\mathtt{Quota}_0^k= \mathtt{HQLR}^k$  which is known as the largest remainder method and $\gamma =1$, which is known as Droop quota. To see how the definitions given here correspond to the many other possible definitions see \cite[Appendix A, \& B]{janson2014asymptotic}.

\subsection{$L^*$ is minimized for a slightly biased apportionment method}
In \Cref{fig:dansk_overrep}, $L^*(\beta)$ is plotted for the Danish election of 2022.
In the regime where the smaller parties are overrepresented regionally, the gap between the upper bound and lower bound in \Cref{prop:rigorous_bound_on_L} is not that tight. Conversely, when the bias of the regional ($\beta$-divisor) method is larger, then the largest party is decisive for the bound on $L^*$ and in that case the interval in \Cref{prop:rigorous_bound_on_L} is tighter.

The analysis shows that since the seat surplus is measured in absolute terms (number of extra seats) fewer compensatory seats are needed to level out in case a large party is overrepresented compared to a smaller party. 
Thus, it is the relative seat surplus that matters for the number of compensatory seats.

\begin{figure}
    \centering
\includegraphics[width=0.7\linewidth]{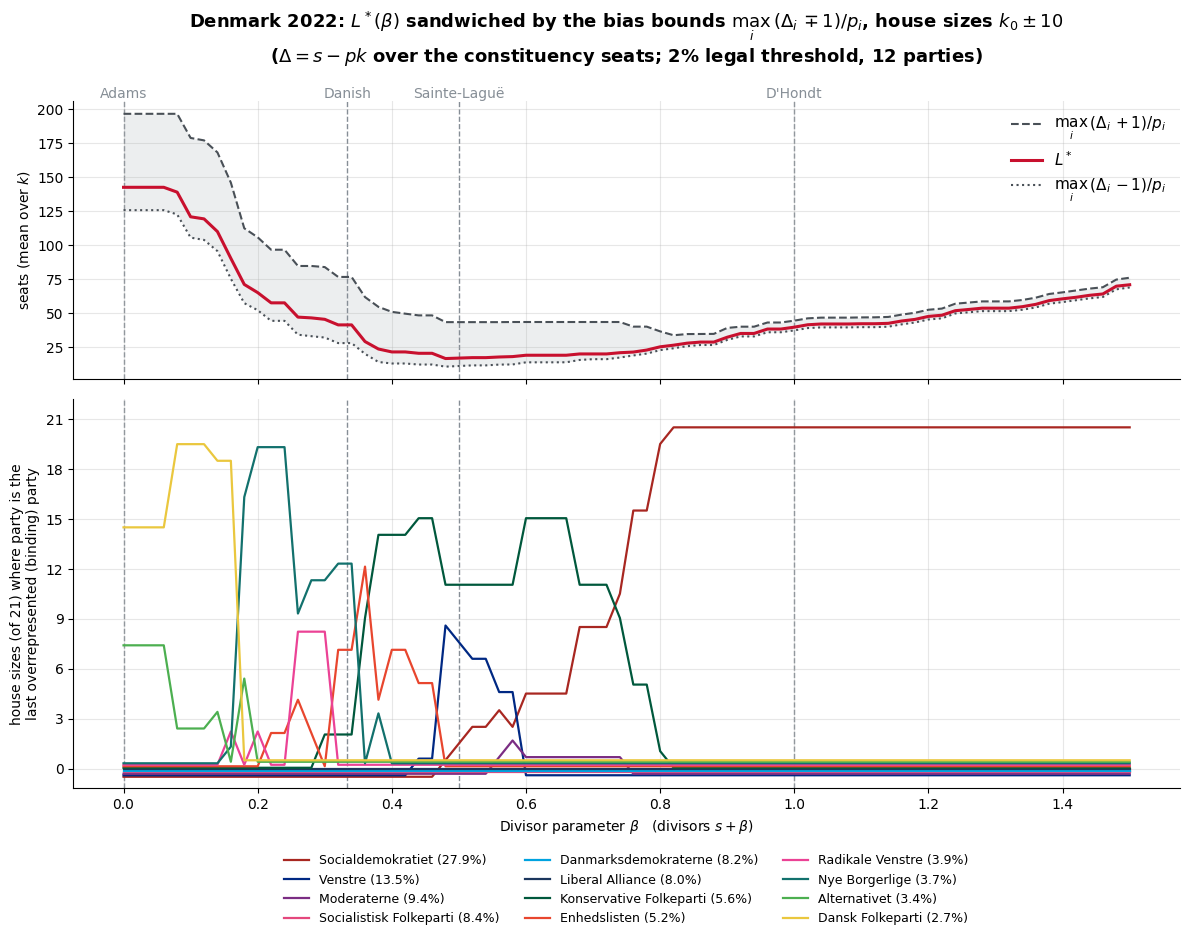}
    \caption{The Danish election of 2022 with regional method in the one-parameter family of $\beta$-divisor methods. For each value of $\beta$ the number of constituency seats is varied from 125 to 145 and apportioned based on the vote totals using the $\mathtt{SL}$-method.  Top) The number of compensatory seats (averaged over the 21 house sizes) vs. the theoretical guarantees given the seat surplus.  Bottom) The anchoring parties and their  percentages (renormalised after the 2\% threshold). When $\beta$ is small, mainly small parties are anchor parties and when $\beta$ is large the (large) party \emph{Socialdemokratiet} is anchoring.  }
    \label{fig:dansk_overrep}
\end{figure}

It is natural to ask for which regional method the system needs the least compensatory seats. 
Intuitively, one might think that using an unbiased method (Sainte-Laguë or Hare Quota) in the constituencies would lead to the least seat surplus accumulation.
But since there are also some fluctuations, the smallest relative seat surplus will usually not occur for the unbiased method, but rather for a method which is slightly biased towards larger parties. 

The $\beta$-divisor methods have an asymptotic bias proportional to $\beta - \frac{1}{2}$, see \eqref{eq:expected_bias} below. 
\Cref{fig:Lstarofbeta} shows that $L^*$ seems to be lowest for a slightly biased method with a $\beta$ of around $0.6$. The picture is less clear for quota methods. 

\begin{figure}
    \centering
    \includegraphics[width=0.48\linewidth]{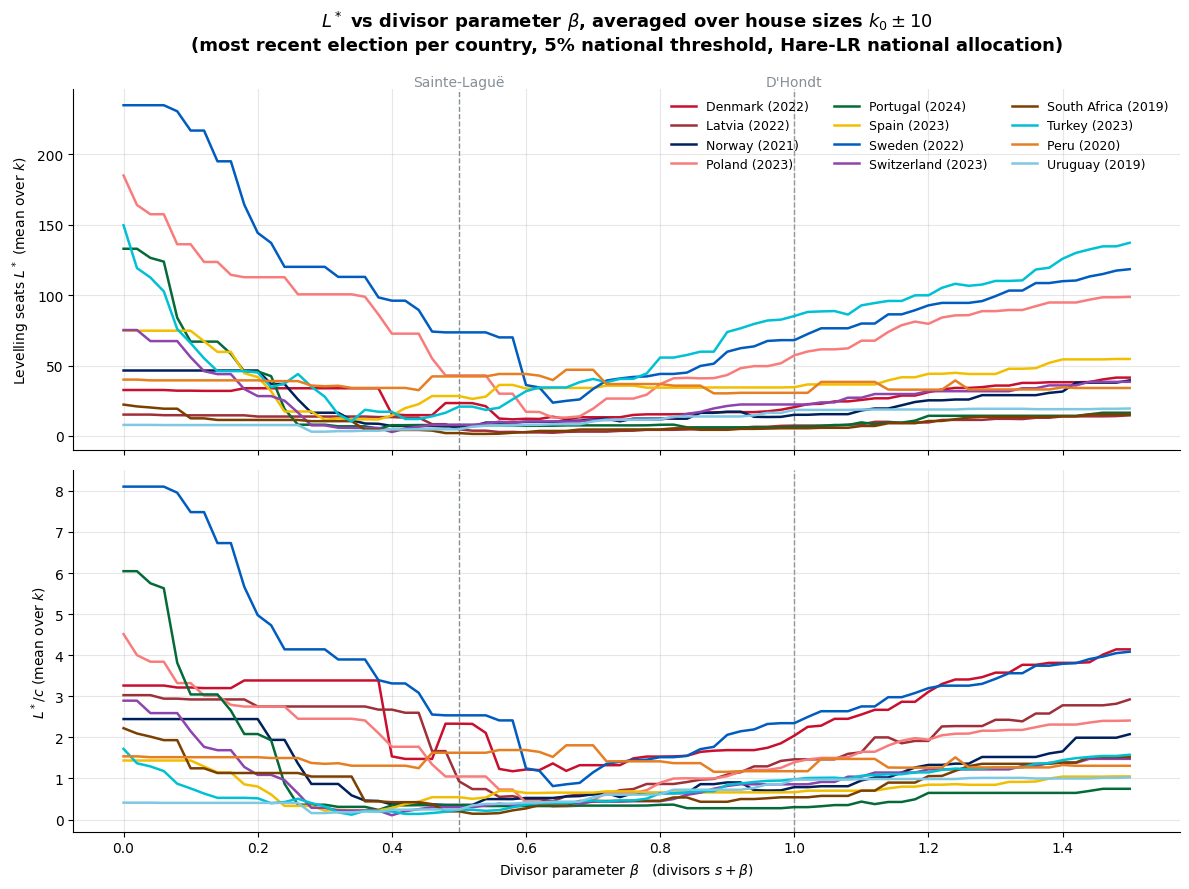}
        \includegraphics[width=0.48\linewidth]{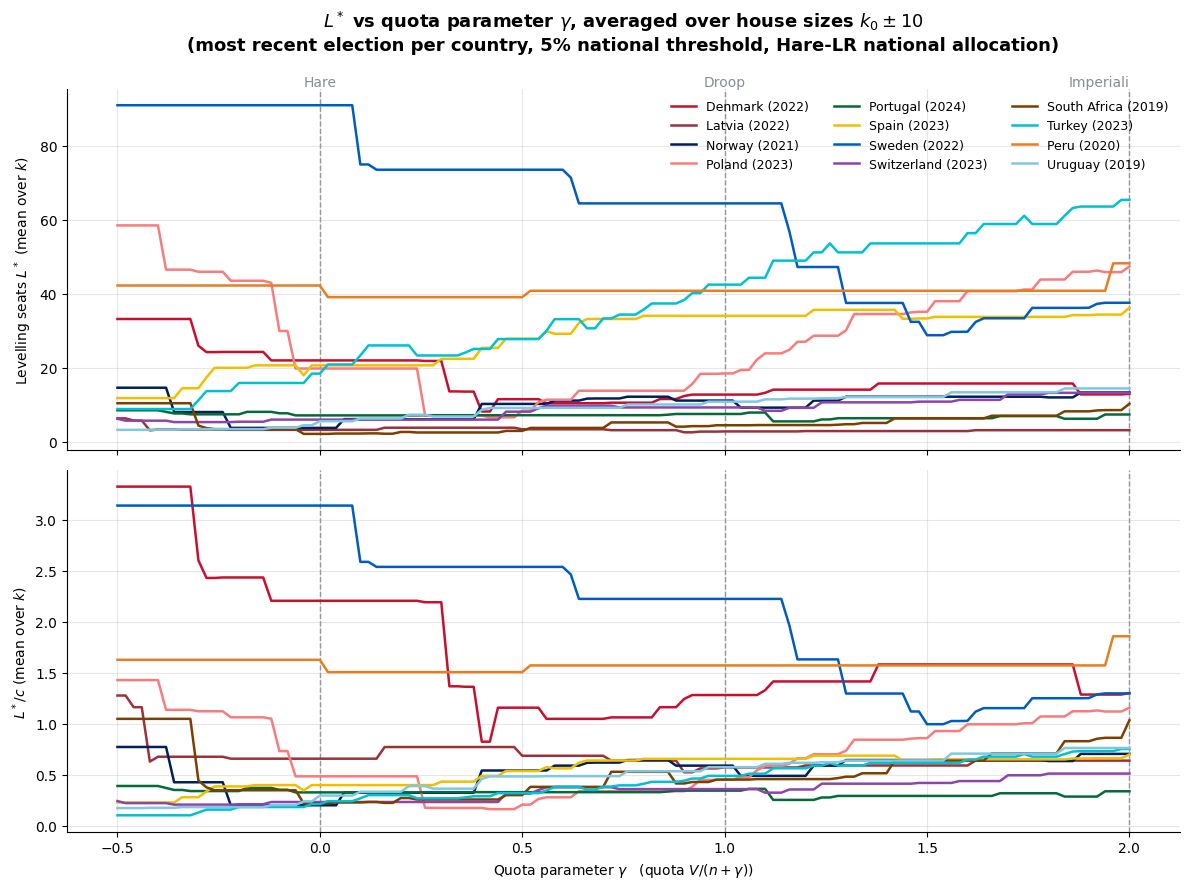}
    \caption{Evaluation of $L^*(\beta)$ and $L^*(\gamma)$ for the most recent election (in the dataset) for 12 countries of interest. The results are average over 21 house sizes around the current number of regional seats.  Note that the numbers are significantly larger for the divisor methods than for quota methods. The data shows that the values of $\beta$ and $\gamma$ that minimize $L^*$ are not the unbiased values, but values that have a (very slight) bias towards larger parties.}
    \label{fig:Lstarofbeta}
\end{figure}

\section{Expected number assuming the asymptotic bias} \label{sec:expected_bias}
The study of the bias of the D'Hondt method started more than 100 years ago \cite{polya1919proportionalwahl} and is still active \cite{schuster2003seat}, \cite{bochsler2010gains}, \cite{janson2014asymptotic}.  

In \cite{flis2020pot}, an empirical seat surplus formula was derived for one-tier systems using the D'Hondt method. We propose the following generalization to $\beta$-divisor methods, 
\begin{align}\label{eq:expected_bias}
    \mathbb{E}[\Delta_i] = nc (\beta- \frac{1}{2})(p_i - \frac{1}{n}). 
\end{align}
This formula, which was empirically derived in \cite{flis2020pot} for $\beta = 1$, can also be derived by summing up the expectations in \cite{janson2014asymptotic}. I.e., the single district biases add up across $c$ constituencies and that the vote share of the party does not vary too much across constituencies. 
For the D'Hondt method, the formula can be proven in specific circumstances \cite{boratyn2025seat} and it is an interesting avenue for future work to pin down the conditions for the general formula as well.

Across electoral systems the accuracy of the seat surplus formula varies. In \cite{flis2020pot} the formula in \eqref{eq:expected_bias} was reviewed empirically for the D'Hondt method $\beta =1$, which is also our main use case for the formula.

For the Danish case, this was assessed for the D'Hondt method in \cite{klausensandsynligheden}. There, it was shown how the formula \eqref{eq:expected_bias} became more precise when the size of the constituencies increased. It was shown how the $\beta =1$ formula fit very well in the Danish case following the electoral reform in 1970, but not before that.  The reason for that was that the rurally strong party Venstre benefited from a distortion of the local seats that gave rural areas larger representation. 

This example highlights how it is always regionally strong parties that are the enemy of the formula, as it was also discussed in  \cite{flis2020pot, boratyn2025seat}. 

If the regions are very heterogeneous such as the Kurdish-dominated Turkish South East  \cite{evci2020regional} or the Copenhagen constituency in Denmark \cite{klausensandsynligheden} it has been discussed how it makes sense to use the formula regionally. 
When discussing the rigorous bounds, we will be concerned with the party performance in each constituency, which is in a sense a further 
fine-graining of this method. 

Further empirical investigation of this formula, especially for other values of $\beta$ than 1 is an interesting avenue for further research that would also strengthen the conclusions of this paper. 

\subsection{The number of required compensatory seats assuming the bias for the D'Hondt method}
In the following, we assume the formula for the seat surplus \eqref{eq:expected_bias} and derive approximately how many compensatory seats are needed to ensure proportionality.
Recall a party is anchoring if it has incongruent seats with $L^* -1$ compensatory seats. 

If the malapportionment is small and the regional method has a bias towards the largest party (e.g., $\beta > \frac{1}{2}$, corresponding to the D'Hondt method at $\beta=1$) the largest party will often be the anchor party.
Thus, it makes sense to consider which $L$ is needed for the largest party to be congruent.

Since the election is assumed to be guaranteed proportional (and the number of seats is typically much larger than the number of seats in each constituency), one would expect that a party with vote share $p$ should obtain around $p (k+L)$ seats in total\footnote{This is of course not the case if a significant proportion of the vote is wasted on an electoral threshold.}.
On the other hand, the party should expect a total number of $pk + \Delta$ constituency seats. The condition for incongruence is having more constituency seats than seats in the national apportionment. Using  \eqref{eq:expected_bias} this condition becomes, 
\begin{align}
  p (k+L) <  pk +  nc (\beta- \frac{1}{2})(p - \frac{1}{n}) 
\end{align}
We expect that the $L$ for which these two expressions are (approximately) equal is  $L^*$, the smallest number of compensatory seats needed to avoid incongruence,  
\begin{align}
    p \mathbb{E}[L^*]= nc (\beta- \frac{1}{2})(p - \frac{1}{n}),
\end{align}
which is again equivalent to
\begin{align}\label{eq:expected_number_of_levelling_seats}
      \boxed{\frac{\mathbb{E}[L^*]}{c}= (\beta- \frac{1}{2})n(1- \frac{1}{np}).}
\end{align}
Recall, that the D'Hondt method $\beta =1$. Here $\frac{L^*}{c}$ denotes the number of compensatory seats per constituency. The number $1-\frac{1}{np}$ can be interpreted as a measure of how much larger the largest party is than the average party. 
This party surplus number increases with the number of parties and with the size of the largest party. We emphasize that it is very different from the effective number of parties \cite{laakso1979effective}, which is usually discussed in the literature. Notice how $k$ disappeared in the calculation\footnote{The number $k$ does play a role implicitly, since for instance if $k=c$, the method is equivalent to first past the post. In that case the assumed formula \Cref{eq:expected_bias} breaks down.}.  Instead the number of constituencies is central. 

Generally, we expect the formula \eqref{eq:expected_bias} to be most precise when $\beta$ is large and when the average constituency size is not too small. This is because only in this case the method bias is the most important factor. For smaller values of $\beta$, as for example the SL method, the fluctuations are dominant.

Even more importantly, we only expect the largest party to be decisive for the incongruence when the method in the constituencies favours the largest party, see for example the crossover in the anchor party (decisive for incongruence ) in \Cref{fig:dansk_overrep} as a function of $\beta$. 

\begin{remark}
    The size of the expected overrepresentation in seats is
$\max(  nc (\beta- \frac{1}{2})(p - \frac{1}{n}) - Lp, 0)$. 
\end{remark}

\subsection{Expected number of required compensatory seats for quota methods}
For a $\gamma$-quota method the expected asymptotic bias in a single constituency was calculated in \cite[Theorem 3.11]{janson2014asymptotic} to be $\gamma(p - \frac{1}{n})$.

When all parties are in play to get local seats in all constituencies, the malapportionment is small, and the party in question has homogeneous support then we expect that the seat surplus accumulates across the $c$ constituencies, in which case 
\begin{align}\label{eq:expected_bias_quota_method}
    \mathbb{E}[\Delta] = c \gamma(p - \frac{1}{n}). 
\end{align}
Assuming overall proportionality yields the following inequality as an indicator of incongruence,
\begin{align*}
  p (k+L) <  pk +  c \gamma(p - \frac{1}{n}). 
\end{align*}
Solving the case of equality yields the number of required compensatory seats 
\begin{align} \label{eq:expected_number_of_levelling_seats_quota}
   \boxed{\frac{\mathbb{E}[L^*]}{c} =  \gamma(1 - \frac{1}{np}).} 
\end{align}

\begin{remark}
    The accuracy of the derived formulas for the expected number of compensatory seats in \eqref{eq:expected_number_of_levelling_seats} and \eqref{eq:expected_number_of_levelling_seats_quota} depends on the party systems and they are not rigorous guarantees. 
    In the unrealistic limit where the number of seats in each constituency $k_j$ all tend to infinity, analogs of \eqref{eq:expected_bias} and \eqref{eq:expected_bias_quota_method} (adjusted for the support of the parties in each constituency) become exact by the proofs of Janson \cite{janson2014asymptotic}. But the fluctuations do matter to obtain the distribution of \Cref{prop:rigorous_bound_on_L}, so therefore we do not expect \eqref{eq:expected_number_of_levelling_seats} and \eqref{eq:expected_number_of_levelling_seats_quota} to become exact in any limit. 
    
    Similarly, the methods of \cite{boratyn2025seat} can potentially be used to derive rigorous statements under another set of assumptions.
\end{remark}

\section{Malapportionment}\label{sec:malapportionment}



This section formalizes the role that \emph{malapportionment} plays in the seat surplus accumulation. 
With other aspects of seat surplus accumulation (cf. \Cref{fig:one-tier_mechanisms})  held fixed, parties that are strong where the seats are cheap will obtain a larger seat surplus than parties that are strong where the seats are expensive. 
It turns out that a correction to the seat surplus formula can be precisely described. 
 The perspective is quite different from the conventional approach of defining an index of apportionment as in \cite{samuels2001value}; it is for example shown how malapportionment is irrelevant for parties with uniform support.

\Cref{tab:seat-apportionment}  shows how the apportionment is accomplished in practice in our chosen set of countries. 
Since the constituency seats are not apportioned with respect to the number of votes in that particular election, it will necessarily be the case that some constituency seats are "cheaper" than others. There is however a general correspondence to either the population or the number of voters, which is of course correlated with the number of votes cast.

The vote share of the $i$th party in the $j$th constituency is denoted $p_i^j$ and the average vote share of the $i$th party $\bar p_i = \frac{1}{c} \sum_{j=1}^c p_i^j.$
Recall, that there are $k_j$ seats and $v_j$ votes in the $j$th constituency. 

Define the \emph{seat discount} $x_j = k_j - \frac{v_j}{v}k$. If $x_j$ is positive it means that the seats in constituency $j$ are "cheaper" than the average constituency. If seats are distributed across constituencies with $\mathtt{HQLR}$ then necessarily $-1 \leq x_j \leq 1$. In the empirical analysis below we use $\mathtt{SL}$ where it is almost always the case.  
Recall, the definition of the covariance $\Cov(p,x) = \mathbb{E}[(p-\mathbb{E}[p])(x-\mathbb{E}[x])]$, here $\mathbb{E}$ is the average over constituencies, e.g., $\mathbb{E}[x] = \frac{1}{c}\sum_{j=1}^c x_j = \frac{1}{c}\sum_{j=1}^c k_j - \frac{v_j}{v}k  = 0$ and $\mathbb{E}[p_i] = \frac{1}{c}\sum_{j=1}^c p_i^j$. 

The following lemma shows that the additional seat surplus that a party gets nationally compared to the summed out surpluses from the constituencies is exactly the covariance of the vote share with the seat discount $x_j$.

\begin{lemma}\label{lemma:covariance_is_difference}
  Label a party with $1$, and suppose it obtains $v_1^j$ votes in the $j$th constituency.  Then\footnote{Using Hölder's inequality this number can be bounded by $\leq \min \left( 2k (\max_j p_1^j)\cdot d_{TV}(\mathcal{K},\mathcal{V}), c \bar p_1 \max_j \abs{k_j - \frac{v_j}{v}k} \right)$.
    In particular, if the regional seats are distributed using a quota method. Then 
    $\abs{ \sum_{j=1}^c \Delta^j_1 - \Delta_1} \leq  c \bar p_1.$}
$$
\Delta_1 - \sum_{j=1}^c \Delta_1^j  = c\Cov(p_1, x).   
$$

\end{lemma}
\begin{proof}
By definition of the local seat surplus and the global seat surplus:
\begin{align*}
    \Delta_1 -\sum_{j=1}^c \Delta_1^j  = 
    s_1 - \frac{v_1}{v}k  -  \sum_{j=1}^c  (s_1^j - \frac{v_1^j}{v_j}k_j)=     - \frac{v_1}{v}k  +  \sum_{j=1}^c  \frac{v_1^j}{v_j}k_j
    =      \sum_{j=1}^c k\frac{v_1^j}{v_j}(\frac{k_j}{k} - \frac{v_j}{v})
        =      \sum_{j=1}^c (p_1^j - \bar p_1)(k_j- k\frac{v_j}{v}),
\end{align*}
where we used in the last step that $\sum_{j=1}^c \bar p_1(k_j- k\frac{v_j}{v})=0$. 
\end{proof}
Now, $c\Cov(p_1, x)$ can be bounded if either the vote shares of party $1$ are quite uniform or if the distribution of seats to constituencies is quite proportional. 
If for instance the constituency seats are distributed with Hare Quota, then $\abs{k_j- k\frac{v_j}{v}} \leq 1$ and the whole expression is bounded by $c \bar p_1.$
On the other hand, if party 1 has uniform geographic support, then $p_1^j = p_1$ and the sum is 0.

\Cref{tab:cov_seat_price} shows how the effects of malapportionment depend heavily on context. On the left, in the Danish case, the values are generally less than 1 seat, albeit larger for the actual area-dependent apportionment (cf. \Cref{tab:seat-apportionment}). On the right, the HDP and YSP with their stronghold in the Turkish South-East where the seats are cheaper win on that account, while the CHP, whose strongholds are the big cities, lose dramatically.  


\begin{table}[t]
\centering
\caption{Apportionment of regional seats across districts. The references have been found using AI and the relevant links are given in the references and the linked sources have been double-checked with the table amended in many circumstances. The * denotes circumstances where it was particularly difficult to obtain the precise information. }
\label{tab:seat-apportionment}
\small
\begin{tabular}{@{}l p{3.9cm} p{2.7cm} p{4.4cm} l@{}}
\toprule
Country & Basis & Method & Extra rules & Ref. \\
\midrule
Denmark  & population $+$ votes at last election $+$ $20\cdot$area (km$^2$) & Hare-LR & recomputed every 5 years; 135 seats over 10 storkredse ($+2$ Faroe Islands, $+2$ Greenland) & \cite{folketingsvalgloven} \\
Norway  & population $+$ $1.8\cdot$area (km$^2$) & Sainte-Lagu\"e & floor of 4.  1 of each district's seats is a compensatory seat (19 of 169) & \cite{valgloven2023} \\
Sweden  & voters (on 1 March of election year) & Hare-LR & 310 fixed seats recomputed before every election; 39 adjustment seats float nationally & \cite{vallag2005} \\
Spain  &  population & Hare-LR & floor of 2 per province; Ceuta and Melilla 1 each & \cite{loreg1985} \\
Portugal  & voters & D'Hondt & $2+2$ emigrant seats (Europe / outside Europe) & \cite{lei14de1979} \\
Poland  & population & Close to Hare-LR & --- & \cite{kodekswyborczy2011} \\
Switzerland* & permanent-resident population & complicated & floor of 1 per canton & \cite{bpr1976} \\
Latvia  &  voters & Hare-LR & fixed 5 months before the election or on election day (snap election) & \cite{saeimaelectionlaw1995} \\
Turkey  & population & complicated & floor of 1 per province; large provinces split up & \cite{turkeylaw2839} \\
Peru  & voters & "proporcional" & floor of 1 per district, max 2 for overseas district & \cite{ley26859} \\
Uruguay* & registered voters & proportional/unclear & constitutional floor of 2 per department & \cite{uruguayconst1967} \\
South Africa* & registered voters & Commission determination & ---   & \cite{saelectoralact1998} \\
\bottomrule
\end{tabular}
\end{table}
\begin{table}[ht]
\centering
\small
\setlength{\tabcolsep}{3.5pt}
\begin{subtable}[t]{0.52\textwidth}
\centering
\begin{tabular}{l rr rr rr}
\toprule
 & \multicolumn{2}{c}{2015} & \multicolumn{2}{c}{2019} & \multicolumn{2}{c}{2022} \\
\cmidrule(lr){2-3} \cmidrule(lr){4-5} \cmidrule(lr){6-7}
Party & actual & SL & actual & SL & actual & SL \\
\midrule
Socialdemokratiet & 0.10 & 0.00 & 0.32 & 0.04 & 0.20 & 0.05 \\
Venstre & 0.20 & -0.10 & 0.32 & 0.00 & 0.13 & 0.00 \\
Dansk Folkeparti & 0.16 & -0.06 & 0.16 & 0.03 & 0.03 & 0.02 \\
Radikale Venstre & -0.12 & 0.04 & -0.31 & -0.02 & -0.11 & -0.02 \\
Moderaterne & -- & -- & -- & -- & -0.10 & 0.02 \\
Socialistisk Folkeparti & -0.05 & 0.02 & -0.16 & -0.01 & -0.09 & -0.01 \\
Danmarksdemokraterne & -- & -- & -- & -- & 0.20 & 0.00 \\
Liberal Alliance & -0.11 & 0.02 & -0.04 & 0.00 & -0.12 & -0.01 \\
Enhedslisten & -0.09 & 0.06 & -0.19 & -0.02 & -0.09 & -0.04 \\
Konservative Folkeparti & -0.05 & -0.01 & -0.12 & -0.01 & -0.04 & -0.01 \\
Alternativet & -0.08 & 0.05 & -0.08 & -0.01 & -0.07 & -0.03 \\
Nye Borgerlige & -- & -- & 0.02 & 0.01 & 0.01 & 0.02 \\
\bottomrule
\end{tabular}
\caption{Denmark (2\% legal threshold).}
\label{tab:cov_seat_price_dk}
\end{subtable}%
\hfill
\begin{subtable}[t]{0.44\textwidth}
\centering

\begin{tabular}{l rr rr rr}
\toprule
 & \multicolumn{2}{c}{2015} & \multicolumn{2}{c}{2018} & \multicolumn{2}{c}{2023} \\
\cmidrule(lr){2-3} \cmidrule(lr){4-5} \cmidrule(lr){6-7}
Party & actual & SL & actual & SL & actual & SL \\
\midrule
AKP & 3.62 & -0.14 & 1.40 & -0.09 & 1.75 & -0.16 \\
CHP & -16.00 & -0.67 & -6.93 & -0.18 & -7.70 & 0.06 \\
MHP & 0.35 & 0.07 & 0.26 & 0.08 & 1.94 & 0.06 \\
HDP  &  11.69 & 0.79 & 7.11 & 0.38 & -- & -- \\
İYİ  & -- & -- & -1.82 & -0.17 & -1.67 & -0.24 \\
YSP & -- & -- & -- & -- & 7.67 & 0.38 \\
\bottomrule
\end{tabular}
\caption{Turkey (parties $\ge 7\%$ nationally).}
\label{tab:cov_seat_price_tr}
\end{subtable}
\caption{Unnormalised covariance $c\Cov(p_i, x) = \sum_j (p_i^j - \bar p_i)(k_j - k v_j / v)$ between local party support and the seat discount. By \Cref{lemma:covariance_is_difference} it equals the gap $\Delta_i - \sum_j \Delta_i^j$ between the national seat surplus and the summed constituency surpluses (in seats), for (a) Denmark and (b) Turkey. Per year, $k_j$ is either the actual one-tier apportionment (for Denmark the 135 seats) or Sainte-Lagu\"{e} on votes cast, for which the covariance is near zero by construction.}
\label{tab:cov_seat_price}
\end{table}

\subsection{Bounds in the single constituency case.}
The following single constituency deterministic bound is useful. 
\begin{theorem}[\cite{kopfermann1991mathematische}, \cite{janson2014asymptotic}] \label{thm:deterministic_bounds}
    For a $\beta$-divisor method, in an election with $n$ parties, the seat surplus $\Delta$ of a party with vote share $p$ is bounded by, 
\begin{align}\label{eq:divisor_method_seat surplus}
   \abs{\Delta - (\beta - \frac{1}{2})(np -1)} \leq \frac{1}{2} + \frac{n-2}{2}p 
\end{align}
For a $\gamma$-quota method\footnote{a sharper bound exists, but the following within-quota bound will suffice}, 
$
    \abs{\Delta - \gamma( p - \frac{1}{n})} \leq 1.
$

\end{theorem}

When $\beta > \frac{1}{2}$ or $\gamma > 0$, the apportionment methods have a bias towards large parties. The aim of the next section is to sum up this formula across constituencies.

\subsection{Deterministic bounds on the seat surplus in the one-tier case.}
The following bound generalizes the corresponding single constituency bound  \cite{kopfermann1991mathematische} to one-tier systems and at the same time it quantifies the discussion in \Cref{sec:seat surplus_accumulation}. 
\begin{proposition} \label{proposition:total_bias_bound}
For $\beta$-divisor methods, for a party with average vote share $\bar p_i$, the total seat surplus $\Delta_i$ is bounded by 
    \begin{align*}
    \abs{\Delta_i - c(n-\frac{1}{\bar p_i})(\beta - \frac{1}{2}) \bar p_i -c\Cov(p_i,x)} \leq \frac{c}{2} + c \bar p_i \frac{n-2}{2}. 
\end{align*}
\end{proposition}
\begin{proof}
Combine \Cref{lemma:covariance_is_difference} with \Cref{thm:deterministic_bounds} and sum the biases. 
\end{proof}
If the election outcome is random, one could believe that the expected value is the middle of the interval it is confined to\footnote{This can be proven in the asymptotic limit in the single constituency case (see \Cref{thm:deterministic_bounds} ). Here it depends on whether the number of relevant parties $n$ is chosen sensibly. 
}. 
 That would lead us to the following heuristic:
\begin{align}
    \boxed{\mathbb{E}[\Delta_i] = \underset{\text{method bias}}{ \underbrace{c(n-\frac{1}{\bar p_i})(\beta - \frac{1}{2}) \bar p_i}} +\underset{\text{malapportionment transfer}}{ \underbrace{c\Cov(p_i,x)}}.} 
\end{align}
If the support of the party is evenly distributed (or the constituencies have similar size) then $\bar p_i \approx p_i$. 
If the party support is furthermore not correlated with the malapportionment of the votes to constituency seats then $\Cov(p_i,x)$ is 0 in expectation and we recover \Cref{eq:expected_bias}. 
Defining the fluctuations $F_i$ through the following relation we get a complete operationalisation of the four mechanisms discussed in \Cref{sec:seat surplus_accumulation}, 
\begin{align}\label{eq:seat_surplus_split_into_mechanisms}
    \Delta_i = \underset{\text{method bias}}{ \underbrace{c(n-\frac{1}{\bar p_i})(\beta - \frac{1}{2}) \bar p_i}} +\underset{\text{malapportionment transfer}}{ \underbrace{c\Cov(p_i,x)}} + \underset{\text{ fluctuations}}{ \underbrace{F_i.}}
\end{align}
In the single-constituency case, the fluctuations of party $i$, denoted here by $F_i$, were characterized in the limit $k_j \to \infty$  in \cite[Thm 3.7, Thm 3.13]{janson2014asymptotic}.
Combined with \Cref{prop:rigorous_bound_on_L} this allows for a model that would give both expectation and variance of $L^*$ for the entire parameter regime of $\beta$-divisor methods and $\gamma$-quota methods.

\subsection{Deterministic bound on congruence for the largest party}
In this subsection, which can be skipped on first reading, a deterministic bound for the congruence of the largest party is given. 
\begin{theorem}[Deterministic bound, simplest case]\label{thm:deterministic_two_tier_simple}
    Suppose that a two-tier system electoral system has a one-tier base with $c$ constituencies. Make the following assumptions on the electoral system:
    \begin{itemize}
        \item  The regional seats are distributed among constituencies based on the vote totals using Hare Quota. 
        \item  The regional seats are apportioned on parties using a $\beta$-divisor method (with $\beta > \frac{1}{2}$) in the first tier.
        \item  The national tier uses Hare Quota.
    \end{itemize}
    If there are $n$ parties, the largest party has national vote share $p_1$ and average constituency vote share $\bar p_1$. The largest party is congruent if the number of compensatory seats per constituency $L/c$, satisfies
\begin{align}\label{eq:largest_party_congruence}
\boxed{\frac{L}{c}\geq  \beta \frac{\bar p_1}{p_1}( n - \frac{1}{\bar p_1}) + \frac{c+1}{c p_1}.}
\end{align}
\end{theorem}
\begin{proof}
Assuming Hare Quota nationally and that those seats are approximately proportional so that if there are $k+L$ seats in total the party with percentage $p_1$ should obtain around $(k+L)p_1$ of them. 
Thus, if the following inequality is satisfied the largest party cannot be incongruent , 
$$
(k+L)p_1 -1 \geq s_1 = \Delta_1 + kp_1. 
$$
Notice how the terms $kp_1$ cancel and we get the following criterion for congruence:
\begin{align}
    Lp_1 \geq \Delta_1+1. 
\end{align}
Now, inserting our rigorous guarantee for the size of the seat surplus $\Delta_1$ from \Cref{proposition:total_bias_bound} and using that the malapportionment transfer (covariance term) is bounded by $c \bar p_1$ because of the first assumption. That yields the following criterion for congruence:
$$
L p_1 \geq c( n-\frac{1}{\bar p_1})(\beta - \frac{1}{2}) \bar p_1 + \frac{c}{2} + \frac{c \bar p_1 n}{2} +1 
$$
rearranging yields \eqref{eq:largest_party_congruence}. 
\end{proof}

We leave the generalization of \Cref{thm:deterministic_two_tier_simple} beyond Hare Quota to future work.

\begin{remark}
    Two special cases are worth highlighting. For the D'Hondt method $\beta =1$ the criterion is
    \begin{align}
\frac{L}{c}\geq \frac{\bar p_1}{p_1}( n - \frac{1}{\bar p_1}) + \frac{c+1}{c p_1}.
\end{align}
While for the SL-method it is $\beta = \frac{1}{2}$ and so $
\frac{L}{c}\geq \frac{1}{2}\frac{\bar p_1}{p_1}( n - \frac{1}{\bar p_1}) + \frac{c+1}{cp_1}.
$

Another special case of interest is when the party support is homogeneous. In that case, $\bar p_1 = p_1$ and thus, 
\begin{align}\label{eq:bound_hom_support}
    \frac{L}{c}\geq  \beta ( n - \frac{1}{p_1}) + \frac{c+1}{c p_1}.
\end{align}
\end{remark}

\section{Comparison to real-world electoral data}\label{sec:empirical_study}
The goal of this section is to test the logically derived formulas \eqref{eq:expected_number_of_levelling_seats} and \eqref{eq:expected_number_of_levelling_seats_quota} for the expected number of needed compensatory seats against electoral data.

The empirical part of this paper is based on the CLEA dataset on national elections with constituency-level granularity \cite{kollman2024clea}. 
 \Cref{tab:most_recent_election_summary} gives an overview of included countries and the diversity of the data for just the most recent election included in the dataset. 
 The countries selected for the analysis all have elections in multi-member constituencies (and they are either one or two-tier systems). They have been selected mainly with an eye to where there has been the most public discussion about including compensatory seats. 
 With some amount of data cleaning, which is discussed in \Cref{sec:further_data_cleaning} the analysis here is easily extended to any parliament that employs multi-member constituencies.

In the analysis there is a general trade-off of uniformity versus  context. One example of national context is the presence of an electoral threshold.     For the bulk analysis in this section a national electoral threshold of $5\%$ is imposed in each election. But in the specific applications in the next section the legal threshold is used with the caveat that if the country has no threshold we will nevertheless impose a threshold of $2\%$ to limit an explosion in the number of parties. This sometimes changes the results a bit, but some regularization is needed since we use the actual number of parties, $n$, and not some effective party number. For instance, in Spanish elections many parties ran for parliament without obtaining a seat and counting those parties would obscure the analysis.

\begin{table}[htbp]
\centering
\caption{Most recent election summary at the actual number of constituency seats with an imposed 5\% threshold and the constituency seats distributed based on the actual votes (using the SL-method). For the overrepresented parties its percentage $p$ and average constituency percentage $\bar p$ is shown in parenthesis after the party name. $L^*$ is computed when the regional seats are apportioned using both $\mathtt{DH}$ and $\mathtt{SL}$, the compensation requirement $L^*$ is computed with respect to national Hare quota. The vote shares are renormalised for parties above threshold. }
\label{tab:most_recent_election_summary}
\resizebox{\textwidth}{!}{%
\begin{tabular}{lrlllrrrrlrlll}
\hline
Country & Year & Largest Party & $p$ (\%) & $\bar{p}$ (\%) & $c$ & $n$ & $k$ & $L^*$ for $\mathtt{DH}$ & Overrepresented (DH) & $L^*$ for $\mathtt{SL}$ & Overrepresented (SL) & $L^*/c$ for $\mathtt{DH}$  & $0.5(n - 1/p)$ \\
\hline
Denmark & 2022 & Socialdemokratiet & 32.4 & 33.2 & 10 & 8 & 135 & 25 & Socialdemokratiet (32.4\%, 33.2\%) & 7 & Danmarksdemokraterne (9.5\%, 9.2\%) & 2.50 & 2.45 \\
\hline
Latvia & 2022 & Jaunā VIENOTĪBA & 25.0 & 22.4  & 5 & 8 & 100 & 8 & Zaļo un Zemnieku savienība (16.4\%, 19.0\%) & 4 & LATVIJA PIRMAJĀ VIETĀ (8.2\%, 7.9\%) & 1.60 & 2.00 \\
\hline
Norway & 2021 & Arbeiderpartiet & 33.1 & 34.0  & 19 & 5 & 150 & 21 & Arbeiderpartiet (33.1\%, 34.0\%) & 8 & Høyre (25.6\%, 22.3\%) & 1.11 & 0.99 \\
\hline
Peru & 2020 & accion popular & 14.9 & 16.4 & 26 & 9 & 130 & 29 & Alianza para el Progreso (11.6\%, 15.4\%) & 40 & Alianza para el Progreso (11.6\%, 15.4\%) & 1.12 & 1.15 \\
\hline
Poland & 2023 & PiS & 36.8 & 38.1 & 41 & 5 & 460 & 59 & PiS (36.8\%, 38.1\%) & 43 & Konfederacja (7.4\%, 7.4\%) & 1.44 & 1.14 \\
\hline
Portugal & 2024 & Aliança Democrática & 36.1 & 36.4 & 22 & 4 & 230 & 7 & CHEGA (22.6\%, 24.3\%) & 11 & CHEGA (22.6\%, 24.3\%) & 0.32 & 0.62 \\
\hline
South Africa & 2019 & ANC & 65.3 & 68.2  & 9 & 3 & 200 & 5 & ANC (65.3\%, 68.2\%) & 1 & Economic Freedom Fighters (11.2\%, 10.0\%) & 0.56 & 0.73 \\
\hline
Spain & 2023 & PP & 36.9 & 38.6 & 52 & 4 & 350 & 34 & PP (36.9\%, 38.6\%) & 28 & PP (36.9\%, 38.6\%) & 0.65 & 0.65 \\
\hline
Sweden & 2022 & Socialdemokraterna & 32.3 & 33.3  & 29 & 7 & 310 & 63 & Socialdemokraterna (32.3\%, 33.3\%) & 44 & Centerpartiet (7.2\%, 7.1\%) & 2.17 & 1.95 \\
\hline
Switzerland & 2023 & SVP & 30.1 & 30.7 & 26 & 6 & 200 & 21 & SVP (30.1\%, 30.7\%) & 6 & SP (20.1\%, 16.6\%) & 0.81 & 1.34 \\
\hline
Turkey & 2023 & AKP & 39.8 & 40.4 & 87 & 5 & 600 & 86 & AKP (39.8\%, 40.4\%) & 20 & AKP (39.8\%, 40.4\%) & 0.99 & 1.24 \\
\hline
Uruguay & 2019 & Partido Frente Amplio & 42.9 & 34.6 & 19 & 4 & 99 & 21 & Partido Nacional (31.4\%, 35.6\%) & 6 & Partido Frente Amplio (42.9\%, 34.6\%) & 1.11 & 0.83 \\
\hline
\end{tabular}
}
\end{table}

\begin{figure}
    \centering
\includegraphics[width=\linewidth]{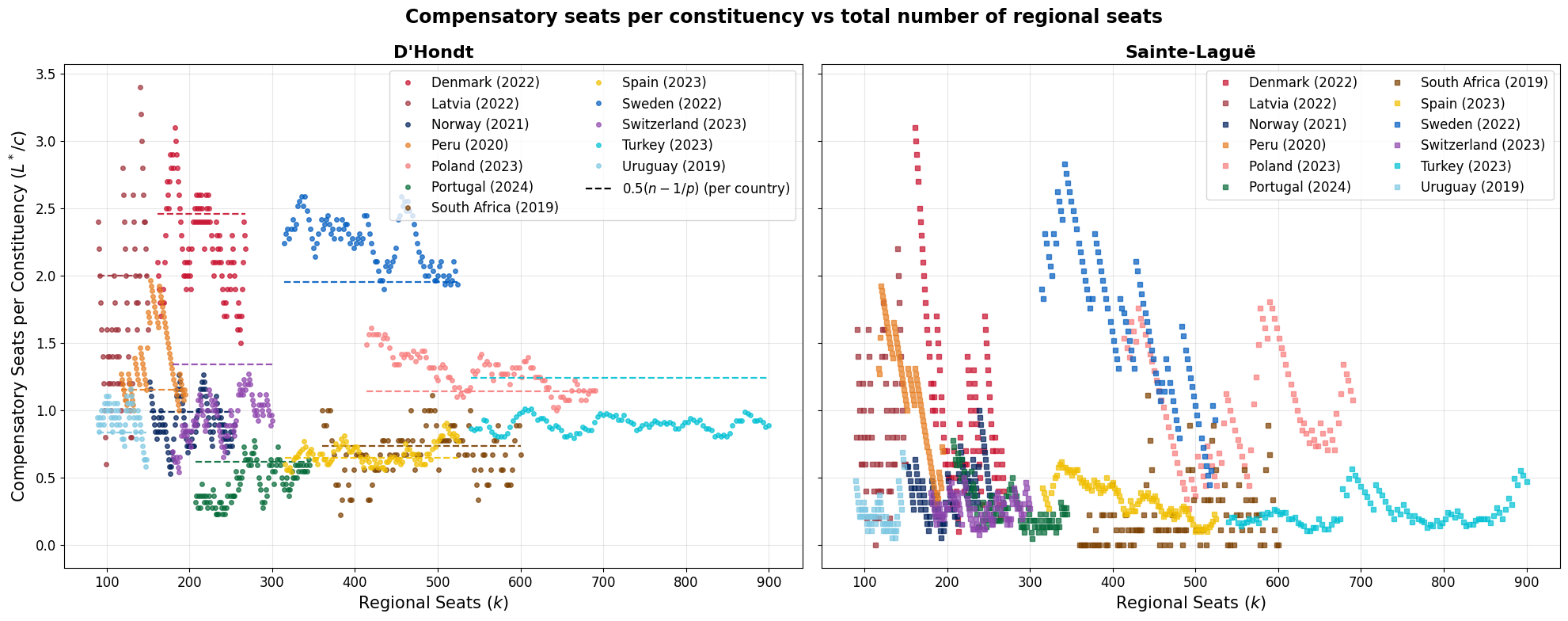}
    \caption{The minimal number of compensatory seats needed for the last election in the CLEA dataset. On the left the D'Hondt method is used to apportion seats to parties in the constituencies, and on the right the Sainte-Laguë method. Naturally, the fluctuations and the degree of discretization are larger for countries with very few constituencies.   }
    \label{fig:varying_parliament_size}
\end{figure}

\begin{figure}
    \centering
    \includegraphics[width=\linewidth]{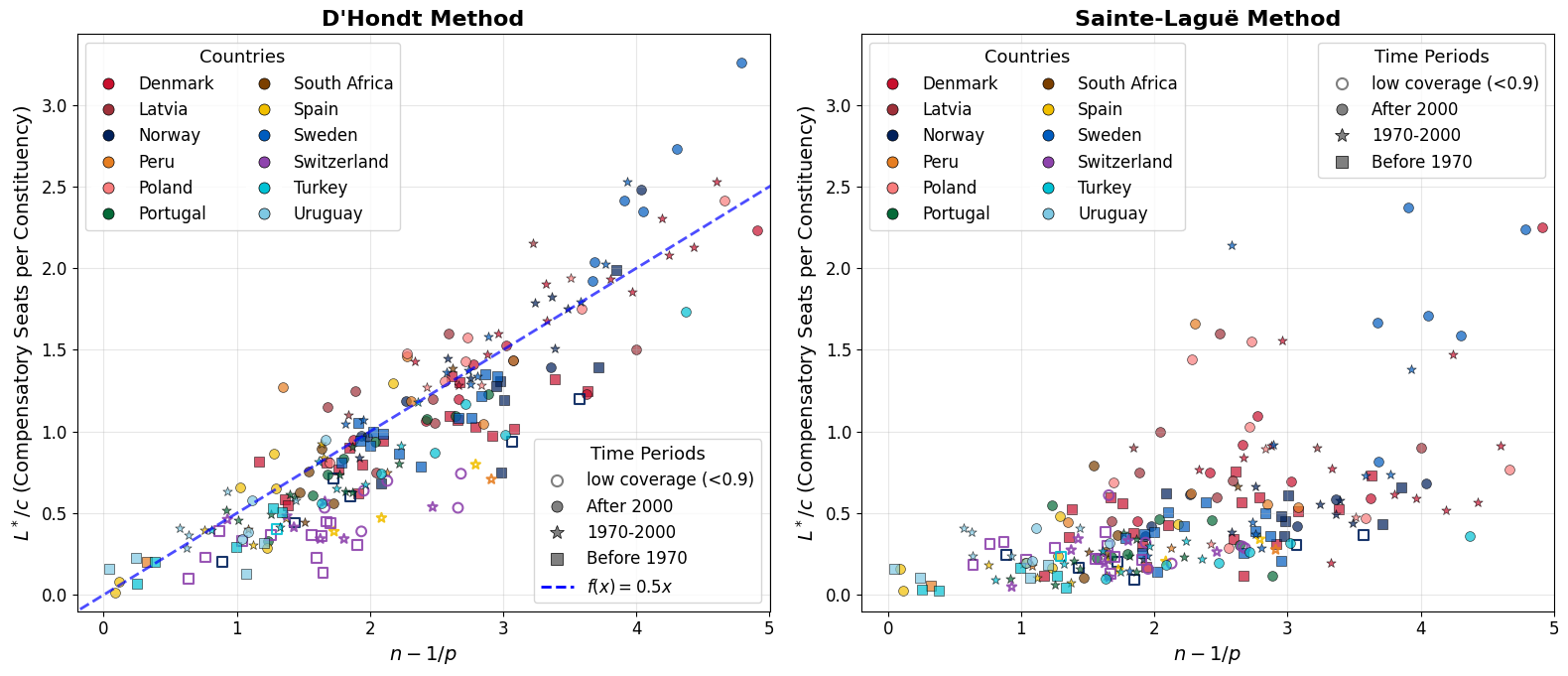}
    \caption{The number of compensatory seats $L^*$ needed per constituency $c$ plotted against the party surplus number $n - \frac{1}{p}$ (the number of parties $n$ above a uniform national threshold at $5\%$, and $p$ is the vote share of the largest party). For every real election based on a one-tier system with total parliament size $k$ the regional seats are distributed using the SL method to constituencies. The parliament size is varied between 0.9 and 1.5 of the actual size and an average is taken. 
    The line corresponds to $f(x) = \frac{1}{2}x$; it is not fitted but logically derived in \eqref{eq:expected_number_of_levelling_seats}. As asymptotically the SL method is unbiased, the same formula is not expected to work for the SL method. \label{figure:numerics_DH_and_SL}}
\end{figure}


\subsection{$L^*$ as a function of $k$}
 Recall that our expectation of the number of needed compensatory seats is independent of the total number of constituency seats $k$.   
We can test that prediction in \Cref{fig:varying_parliament_size}.  Here, for different values of the total number of constituency seats $k$, the regional seats are apportioned based on the number of votes in the election that year using the SL-method, see also \Cref{alg:seats_based_on_vote_totals}. 
Using the resulting regional seat distribution, $L^*$ is computed using \Cref{alg:levelling}. 

In practice, the $k$ regional seats are only apportioned to constituencies for the actual $k$ in use. Thus, to vary $k$ it is necessary to reapportion the regional seats. Throughout this is done using the SL-method. This also mitigates the problem of finding precise data  on the number of seats per constituency.  This simplification eliminates any systematic malapportionment transfer, cf. \eqref{eq:seat_surplus_split_into_mechanisms}.

\Cref{fig:varying_parliament_size} shows that while $L^*$ does change when $k$ varies, as expected the dependence is generally not systematic.
One thing to note on the left of \Cref{fig:varying_parliament_size} is that the total number of constituencies varies wildly across countries (cf. \Cref{tab:most_recent_election_summary}). Generally, the variation for countries with few constituencies is larger than for countries with many constituencies. 

On the right of \Cref{fig:varying_parliament_size} the same is done for the Sainte-Lagu\"e method. 
In that case the variations are generally larger, especially for the countries with few constituencies.


\begin{algorithm}[t]
\caption{Seats distributed to constituencies based on vote totals}\label{alg:seats_based_on_vote_totals}
\begin{algorithmic}[1]
\Input Vote totals in each constituency, $V= (v_1, \dots, v_c)$, number of one-tier seats $k$,
       constituency allocation method $\mathtt{M}$. Throughout $\mathtt{M} = \mathtt{SL}.$
\Output Constituency magnitudes $(k_1,\dots,k_c) = \mathtt{M}^k(v_1, \dots, v_c)$
\end{algorithmic}
\end{algorithm}

\subsection{Divisor methods}
The observation that $L^*$ does not depend on $k$ can be used to reduce the variance in the number of compensatory seats needed for each election by varying $k$.  For each election (country, year) in the data set and each assembly size $k$ between 0.9 and 1.5 of the current assembly size, the constituency seats are apportioned (based on the votes using the SL-method) to obtain ($k_1$, \dots,  $k_c$).
For each such apportionment the minimal number of compensatory seats $L^*(k)$ is calculated. Taking the average over $k$ yields $L^*$ for that election  (corresponding to one data-point in the figure). The results are shown for divisor methods in \Cref{figure:numerics_DH_and_SL} and for quota methods in \Cref{fig:quota_methods}. 
 
For the D'Hondt method (shown on the left in \Cref{figure:numerics_DH_and_SL}) the overall trend is that the number of needed compensatory seats per constituency $L^*/c$ increases linearly with the party fragmentation quantity $n-\frac{1}{p}$ as predicted by \eqref{eq:expected_number_of_levelling_seats}. Of the four regional methods discussed in \Cref{figure:numerics_DH_and_SL} and  \Cref{fig:quota_methods} the data collapse for the D'Hondt method is most remarkable. 

That this is the case is in line with our theoretical expectation since for the four apportionment methods studied the D'Hondt method has the strongest bias (towards the larger parties). This means that in many cases (cf. \Cref{tab:most_recent_election_summary}) the largest party is anchoring. For the SL-method, the anchoring party is more often determined by statistical fluctuations.

The outliers tend to be below the line predicted by logical analysis rather than above the line. 
In some cases, this can be explained by some parties not running in all constituencies, we say the election has  \emph{low coverage}, see \Cref{subsec:low_coverage} for further discussion.  
On the other hand, if there is not a single largest party, the fluctuations of all the largest parties have to be accounted for, which means that additional compensatory seats are needed.
This we suspect might be the case in some of the elections in Sweden that are significantly above the line in  \Cref{figure:numerics_DH_and_SL}.

On the right of \Cref{figure:numerics_DH_and_SL} the same analysis is done, but where the Sainte-Lagu\"e method is used for apportioning constituency seats. Comparing to the left hand side shows that generally fewer compensatory seats are needed (note that the axes are the same).
There is however still a trend that shows that the number of compensatory seats per constituency tends to increase with the party fragmentation $n-\frac{1}{p}$.

\subsection{Quota methods}
For quota methods, the results are plotted in \Cref{fig:quota_methods}. On the left, we again see the expected trend for the Droop method, which favours larger parties, but the trend is less clear. The reason for that is that the bias of the quota methods is smaller. On the right, the same analysis is carried out for Hare Quota, where the variations are even larger and no clear pattern emerges.

\begin{figure}
    \centering
    \includegraphics[width=\linewidth]{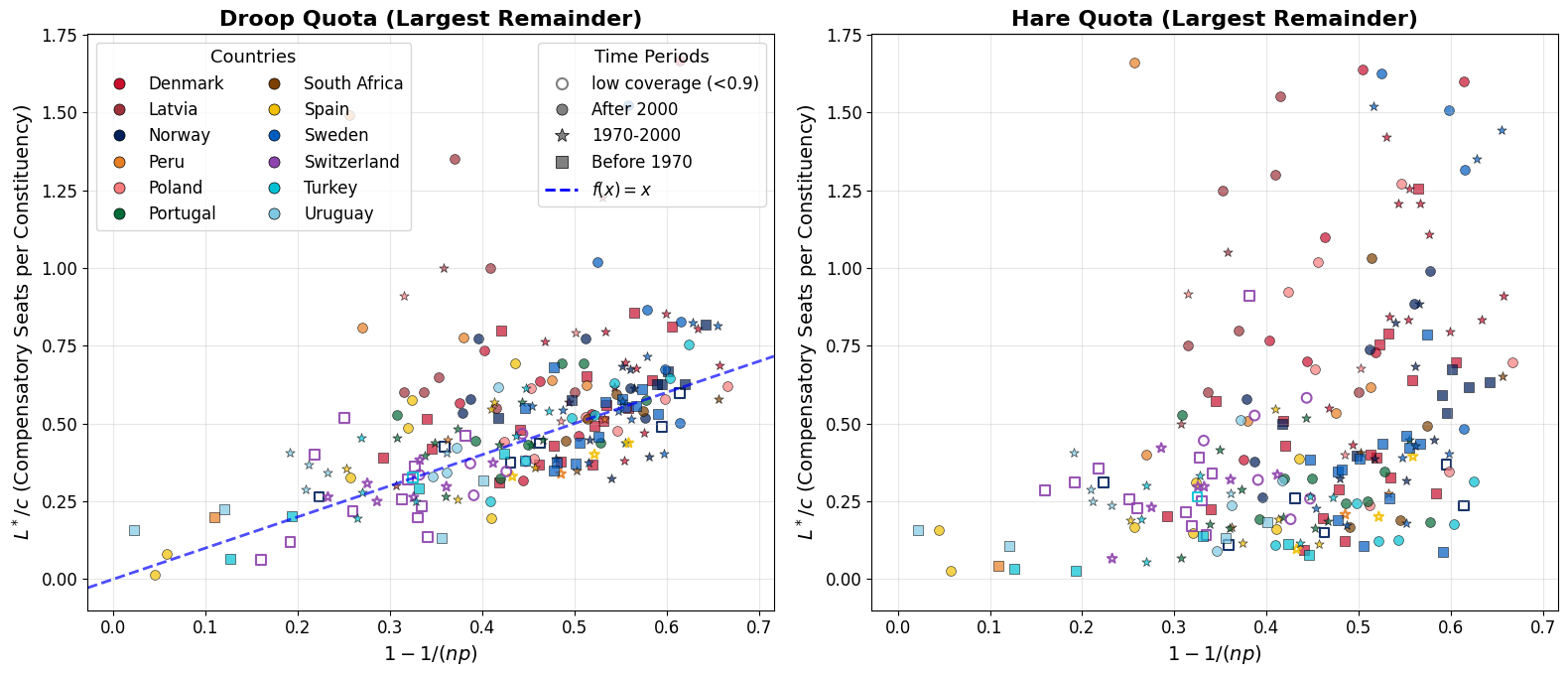}
    \caption{The number of compensatory seats per constituency for the two most common quota methods. The Droop Quota, which is a cousin of the D'Hondt method, and Hare Quota - a cousin of the SL-method. 
    For Droop Quota the expectation of $L^*/c$ is given by \eqref{eq:expected_number_of_levelling_seats_quota} and shown with the line. 
    Notice how the $x$-axis differs compared to \Cref{figure:numerics_DH_and_SL}. }
    \label{fig:quota_methods}
\end{figure}

\subsection{Low coverage}\label{subsec:low_coverage}
Even after imposing a uniform $5\%$ electoral threshold and getting rid of malapportionment by using SL based on vote totals to distribute constituency seats there can be aspects of the local electoral context which are not captured by the theory. 
For example, the framework does not accommodate that some parties are de facto not running in all constituencies.
To define that, we say that an election has \emph{low coverage} if parties run in less than $90\%$ of the constituencies on average. Low coverage means that the number of parties (even above threshold) is artificially inflated and thus we would expect the points in \Cref{figure:numerics_DH_and_SL}left) to be below the line.
This effect is most prevalent in Switzerland, where many parties in the 1950s and 1960s did not run in all constituencies. 


\subsection{Further aspects of data-cleaning}\label{sec:further_data_cleaning}
The CLEA dataset was cleaned in various ways. For instance, many different versions of "Others" are deleted (even though in a few cases others correspond to elected parties, e.g., the Danish party "Retsforbundet" in the elections of 1977 and 1979). 

Another aspect of data cleaning is that local versions of the same party have many different names in the CLEA dataset. 
Parts of this only concern labels, e.g., "svp" or "SVP". 
But it can always be discussed whether two branches of a party are actually the same or not. For example, it seems sensible that the German CDU and CSU should be treated as different parties, whereas it makes most sense in this context to treat versions of the Spanish socialists in different regions as the same party.

In addition, there is a data issue with Swiss cantons where voters have as many votes as the number of representatives from that canton. Since the constituency seats were distributed based on the number of votes this would skew the seat distribution in \Cref{alg:seats_based_on_vote_totals} dramatically. Therefore, the number of votes is divided by the number of votes per voter.

\section{Estimates on the compensation requirement}
Next we discuss the trade-offs that have to be made when choosing a reasonable number of compensatory seats $L$. 


\subsection{Policy advice: how many compensatory seats are sensible?}\label{sec:policy_advice}
Of course no-one knows how a specific party system develops in the future. Especially as the addition of compensatory seats might increase the representation of smaller parties. This again increases the potential momentum of smaller parties and can lead to a more fractured party system, which in turn requires additional compensatory seats.
Therefore, it is hard to say which $n$ and $p$ should be used as a starting point in the formulas.

Furthermore, the Swedish and German examples show that an electoral system can guarantee proportionality (for parties above the electoral threshold) using only a modest number of compensatory seats. The disadvantage with this system is however that parties do not necessarily keep locally obtained seats. 

Thus, in the design of two-tier electoral systems and in the choice of the number of compensatory seats $L$ more specifically, one should not aim for a number so large that a mathematical guarantee can be given.
Instead, one could adhere to the paradigm of guaranteed proportionality and then ask the question of how many compensatory seats are needed to ensure that all locally elected representatives get elected to parliament\footnote{The 3-tier design of the German system means that this number was much higher in the 2025 German election than what was needed mathematically.}.

Since many of the formulas derived depend on the combination $n - \frac{1}{p}$, we need to use a specific value for that number. Two natural estimates are obtained by taking either the maximal over recent elections of 
$ n - \frac{1}{p},$ or to take the even larger pointwise maximum $\max(n) - \frac{1}{\max(p)}$, where the maximum is also over recent elections. Using the former reflects the party structure the best. For the formulas that have terms with $\frac{1}{p}$ added, we use the pointwise min for $p$. 

\definecolor{propblue}{HTML}{0072B2}   
\definecolor{costverm}{HTML}{D55E00}   
\definecolor{inkgray}{HTML}{4A4A4A}
\definecolor{gridgray}{HTML}{9A9A9A}

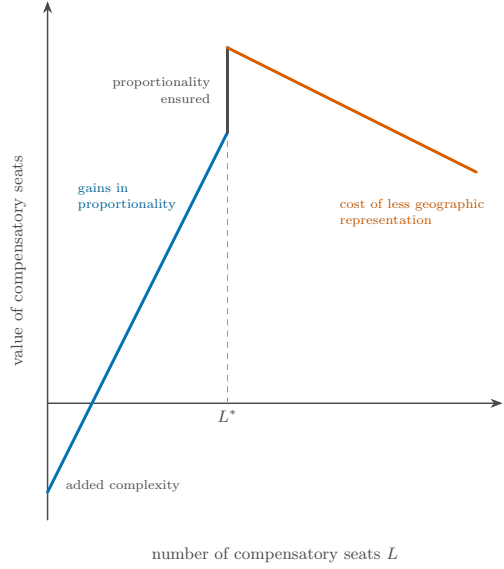
\begin{figure}
\scalebox{0.7}{
\begin{tikzpicture}[>=Stealth, line cap=round, font=\small]
 
\def\Ls{3.4}                       
\def\off{1.68}                     
\pgfmathsetmacro{\yA}{6.8-\off}    
\pgfmathsetmacro{\yB}{8.4-\off}    
\pgfmathsetmacro{\xend}{8.1}
\pgfmathsetmacro{\yend}{\yB-0.5*(\xend-\Ls)}   
\def\xmax{8.6}
\def\ytop{7.6}
\def\ybot{-2.2}
 
\draw[->, thick, inkgray] (0,0) -- (\xmax,0);
\node[below, text=inkgray] at (\xmax/2,-2.6) {number of compensatory seats $L$};
\draw[->, thick, inkgray] (0,\ybot) -- (0,\ytop);
\node[rotate=90, text=inkgray] at (-0.55,{(\ytop+\ybot)/2}) {value of compensatory seats};
 
\draw[dashed, gridgray] (\Ls,0) -- (\Ls,\yA);
\node[below, text=inkgray] at (\Ls,0) {$L^{*}$};
 
\draw[propblue, line width=1.6pt] (0,-\off) -- (\Ls,\yA);
\draw[inkgray,  line width=1.6pt] (\Ls,\yA) -- (\Ls,\yB);
\draw[costverm, line width=1.6pt] (\Ls,\yB) -- (\xend,\yend);
 
\node[inkgray, align=left, anchor=west] at (0.22,-1.55)
      {\scriptsize added complexity};
\node[propblue!85!black, align=left, anchor=west] at (0.45,3.9)
      {\scriptsize gains in\\[-2pt]\scriptsize proportionality};
\node[inkgray, align=right, anchor=east] at (\Ls-0.2,{\yB-0.8})
      {\scriptsize proportionality\\[-2pt]\scriptsize ensured};
\node[costverm!90!black, align=left, anchor=west] at (5.4,3.6)
      {\scriptsize cost of less geographic\\[-2pt]\scriptsize representation};
 
\end{tikzpicture}}
    \caption{Schematic illustration of the value of compensatory seats for a specific election assuming that proportionality, geographic representation and a bound on the parliament size all matter. Exactly at $L^*$ there is a jump which is both because the algorithm necessarily has to be more complicated below $L^*$ and because either guaranteed proportionality or geographical representation can no longer be satisfied. The reason the curve starts below the $x$-axis is the added complexity of the two-tier electoral system. }
    \label{fig:schematic_value_of_levelling_seats}
\end{figure}

\subsection{The choice of method in the constituencies}
Generally, more compensatory seats are needed with the D'Hondt method than with methods that have less bias. This has been demonstrated using many different methods in almost every section of this paper.

In addition, constituency seats sometimes act as back doors for electoral threshold. To make it harder for a smaller party to obtain the first seat in a constituency, the first SL-divisor has been modified in various ways. The main argument against quota methods is the well-known lack of population monotonicity.

\subsection{Disadvantages of compensatory seats and alternatives to two-tier systems}
There are several complexity costs associated with the addition of compensatory seats.
Two-tier electoral systems are necessarily more complicated than one-tier systems, in particular because of the necessary incongruence-breaking mechanism that has to exist. 

On the other hand, in one-tier systems the results can be evaluated locally; without any involvement of say a central government. The absence of national lists means that the entire election including ballots can be held regionally.  This structure arguably means that the power distribution within parties is more regionalized.

If an electoral system imposes compensatory seats, then there are several reasons not to have an excessive number.
One could argue that increasing the number of compensatory seats weakens the geographical representation of the system. The compensatory seats are however often assigned to a constituency. 
It has also been mentioned that politicians sometimes view compensatory seats differently from constituency seats. 
Electoral systems such as the former German system, which adds compensatory seats to parliament (ge.\ \emph{Ausgleichsmandate}, used to compensate for additional incongruent seats overhang seats)  until congruence is ensured, do not have any bounds on the number of seats in parliament, which comes with clear problems of cost and public support.  

To sum up, there might be some arguments against having far too many compensatory seats. A schematic summary of the arguments is shown in \Cref{fig:schematic_value_of_levelling_seats}. 
The figure illustrates that the marginal value of a compensatory seat is the same until ensured proportionality is reached (at $L^*$ compensatory seats).

For each election outcome each added compensatory seat counts towards reducing disproportionality until $L^*$ is reached whereafter the disproportionality is not further reduced. 
Policy makers constrained by a fixed parliament size should aim for a number $L$ which is expected to be just larger than $L^*$. 

Since the effects cannot be quantified and the number $L^*$ depends on the party structure, at best an informed assessment can be made when judging the number of needed compensatory seats. 
Such country specific estimates are made in the next subsection.

\begin{table}[htbp]
  \centering
  \small
  \renewcommand{\arraystretch}{1.2}
  \setlength{\tabcolsep}{2pt}
\setlength{\extrarowheight}{3pt}
{\renewcommand{\arraystretch}{1.15}%
\begin{tabular}{@{}llccc ccc cc@{}}
    \hline
    \multicolumn{10}{@{}l}{Estimate on the compensatory requirement $L^*$ (max over the 5 most recent elections, legal $\tau$)} \\[2pt]
    \hline
    & Method & $L$ in use & $\tau$ (\%) & $k+L$
      & {\small $c\bigl(\beta-\frac12\bigr)\bigl(n-\frac1p\bigr)$}
      & {\small $\max L^*$} & {\small $\max L^*$}
      & {\small $c\beta\bigl(n-\frac{1}{p_1}\bigr)+\frac{c+1}{p_1}$}
      & {\small $c\beta\frac{\bar p}{p}\bigl(n-\frac{1}{\bar p}\bigr)+\frac{c+1}{p}$} \\
    & & & & & \eqref{eq:expected_number_of_levelling_seats} & realized & $k\pm5$ & \eqref{eq:bound_hom_support} & \Cref{thm:deterministic_two_tier_simple} \\
    & & & & & expected & 5 elect. & 5 elect. & guarantee & guarantee \\[2pt]
    \hline
    \textcolor[HTML]{C8102E}{\textbf{Denmark}} & \ovalbox{$\mathtt{DH}$} & 40 & 2 & 175 & 42 & 39 & 45 & 124 & 143 \\
     & $\mathtt{SL}$ &  &  &  & -- & 33 & 38 & 81 & 94 \\[2pt]
    \hline
    \textcolor[HTML]{9E3039}{\textbf{Latvia}} & $\mathtt{DH}$ & 0 & 5 & 100 & 10 & 9 & 12 & 44 & 40 \\
     & \ovalbox{$\mathtt{SL}$} &  &  &  & -- & 5 & 8 & 34 & 33 \\[2pt]
    \hline
    \textcolor[HTML]{00205B}{\textbf{Norway}} & $\mathtt{DH}$ & 19 & 4 & 169 & 38 & 39 & 45 & 136 & 142 \\
     & \ovalbox{$\mathtt{SL}^\dag$} &  &  &  & -- & 19 & 27 & 102 & 105 \\[2pt]
    \hline
    \textcolor[HTML]{F87C7C}{\textbf{Poland}} & \ovalbox{$\mathtt{DH}$} & 0 & 5 & 460 & 73 & 75 & 77 & 248 & 249 \\
     & $\mathtt{SL}$ &  &  &  & -- & 58 & 63 & 175 & 175 \\[2pt]
    \hline
    \textcolor[HTML]{046A38}{\textbf{Portugal}} & \ovalbox{$\mathtt{DH}$} & 0 & 2 & 230 & 53 & 37 & 40 & 179 & 185 \\
     & $\mathtt{SL}$ &  &  &  & -- & 34 & 37 & 126 & 130 \\[2pt]
    \hline
    \textcolor[HTML]{F1BF00}{\textbf{Spain}} & \ovalbox{$\mathtt{DH}$} & 0 & 2 & 350 & 98 & 69 & 71 & 367 & 409 \\
     & $\mathtt{SL}$ &  &  &  & -- & 28 & 30 & 269 & 290 \\[2pt]
    \hline
    \textcolor[HTML]{005CBF}{\textbf{Sweden}} & $\mathtt{DH}$ & 39 & 4 & 349 & 71 & 91 & 92 & 235 & 253 \\
     & \ovalbox{$\mathtt{SL}^\dag$} &  &  &  & -- & 74 & 87 & 170 & 176 \\[2pt]
    \hline
    \textcolor[HTML]{8E44AD}{\textbf{Switzerland}} & \ovalbox{$\mathtt{DH}$} & 0 & 2 & 200 & 50 & 30 & 30 & 185 & 198 \\
     & $\mathtt{SL}$ &  &  &  & -- & 12 & 17 & 139 & 142 \\[2pt]
    \hline
    \textcolor[HTML]{7B3F00}{\textbf{South Africa}} & $\mathtt{DH}$ & 200 & 2 & 400 & 16 & 12 & 14 & 47 & 49 \\
     & $\mathtt{SL}$ &  &  &  & -- & 12 & 16 & 31 & 33 \\[2pt]
    \hline
    \textcolor[HTML]{00C1D4}{\textbf{Turkey}} & \ovalbox{$\mathtt{DH}$} & 0 & 7 & 600 & 118 & 104 & 108 & 438 & 451 \\
     & $\mathtt{SL}$ &  &  &  & -- & 28 & 31 & 329 & 333 \\[2pt]
    \hline
\end{tabular}}
  \caption{Overview of the number of compensatory seats required to avoid incongruence calculated in several different ways. 
  The five rightmost columns are five different "estimates" of how many compensatory seats are needed to avoid incongruence. 
  First, using the formula for the seat surplus of the D'Hondt method. Furthermore, $L^*$ is computed numerically corresponding remarkably well to the formula \eqref{eq:expected_number_of_levelling_seats}. In all cases, the reported number is the maximal number for the five most recent elections in the dataset. 
  The two rightmost columns contain some numbers that have the flavour of rigorous guarantees. They are generally far above the actual numbers. 
 $\mathtt{SL}^\dag$ refers to the modified SL-method and the circle indicates which method is in use.  \label{tab:overview_of_needed_levelling_seats}
 } 
  \label{tab:formulas}
\end{table}

\subsection{Country specific estimates}
The arguments above show that the question of how many compensatory seats should be added to ensure proportionality is a bit ill-posed. At best, a trade-off as the one sketched in  \Cref{fig:schematic_value_of_levelling_seats} should be taken into consideration, while bearing in mind some uncertain estimates of the party system and structure in the future.  
For concreteness, let us nevertheless mention some numbers for some of the countries discussed throughout this paper.

In \Cref{tab:overview_of_needed_levelling_seats}  different estimates on $L^*$ are given, including the largest number the formulas  \eqref{eq:expected_number_of_levelling_seats} and \eqref{eq:expected_number_of_levelling_seats_quota} produced in the worst-case for the last 5 elections (and the same number under small variations of the total number of regional seats $k$). The method in use is also indicated.
Based on the table it is possible to get a rough estimate of the number of needed compensatory seats. For instance the 40 Danish compensatory seats seem reasonable and the 200 South African appear to be far more than needed. Although the seats are to be measured on a per constituency basis, i.e., halving the number of constituencies should halve the corresponding $L^*$, we stick to absolute numbers in the following.


Let us briefly discuss Poland, Portugal, Spain, Switzerland and Turkey, which all currently employ a one-tier electoral system with the D'Hondt method.
In  Poland it has been proposed to move to a mixed system \cite{flis2025mixed} that would increase proportionality. Likewise, the Polish system for the European parliament is two-tier and guaranteed proportional \cite{Pukelsheim2024}. 
To get that using compensatory seats a sensible number would be 75, which is close to the 69 non-compensatory  second tier seats that were in use in the Sejm between 1991 and 2001 \cite{millard2003elections}.  

For Portugal, it seems like around 40 seats would suffice. While 75 is a reasonable number in Spain, which has had a one-tier system with D'Hondt since the transition from dictatorship \cite{riera2017attempts}. There have been several proposals to add a compensatory  second tier \cite{marquez1998spanish,riera2017attempts}.  

Discussions for federal Swiss elections have mainly revolved around changing the electoral method from D'Hondt (known there as Hagenbach-Bischoff) to the  Sainte-Laguë method or around Doppelproporz systems. 
In the current Turkish system, it would be sensible to have 120 compensatory seats. 
In the Turkish election of 1965, a second tier was in use, but it was abolished after that election \cite{teknaz2023adalet}.

\section{Conclusion and outlook}
As discussed, seat surplus can accumulate in one-tier electoral systems in several different ways, but the regional apportionment method is the most important factor.  
In most cases, the widely used D'Hondt method has a much larger seat surplus than Sainte-Laguë or any of the quota methods. 

The main mechanisms of seat surplus are method bias, electoral thresholds, fluctuations, and malapportionment. It was shown that overrepresentation as a result of malapportionment occurs only if the seats in some constituencies are discounted and some parties have strongholds in those constituencies. In the rest of the analysis, the seats were distributed to constituencies using the SL method based on the vote total which eliminates malapportionment. 

The remaining seat surplus was estimated rigorously, in expectation, and by looking at empirical data. Each of the three perspectives shows that the seat surplus scales with the number of constituencies. 
Making the electoral system two-tier by adding compensatory seats can eliminate the seat surplus accumulated across constituencies. 
The number of compensatory seats needed for the electoral result to be congruent with a national apportionment algorithm was denoted $L^*$.
Given the outcome of a one-tier electoral system, it was shown in \Cref{prop:rigorous_bound_on_L} how $L^*$ was determined more or less by the largest relative seat surplus.

A main finding of this paper was that it is most relevant to discuss the number of compensatory seats per constituency, $L^*/c$ and that this number is expected to be linear in the party surplus number $n-\frac{1}{p}$, where $n$ is the number of relevant parties and $p$ is the vote share of the largest party. In particular, this number is independent of the total number of regional seats $k$. This allowed us to reduce the variance of the problem by artificially varying $k$.

 Across the one-parameter family of $\beta$-divisor methods the Sainte-Laguë method is the unique unbiased method (asymptotically), but when constituencies are small the larger parties do tend to get some advantage. The $\beta$ that most often minimizes the number of needed compensatory seats $L^*$ is the slightly biased one with $\beta$ around $0.6$, as opposed to the SL case of $\beta = 0.5$ and D'Hondt of $\beta = 1$. 

We presented some country specific estimates that show that the formulas for the expected required $L^*$ give much more sensible numbers than what one can rigorously guarantee. While the focus of this paper was the number of compensatory seats, designers of electoral systems must have many other aspects in mind (including electoral system, constituencies, their apportionment, voting rules, etc.).  

\subsection{Guaranteed proportionality}
To finish, let us argue for  
\emph{guaranteed proportionality} in multi-member constituencies. By definition, an electoral system is guaranteed proportional if the national allocation has priority (for parties above electoral thresholds). 
One virtue of guaranteed proportionality is that it is hard to game. The national apportionment is easy to understand for voters, while there are many caveats with what is today known as proportional representation. 

Guaranteed proportionality can be achieved without adding enough compensatory seats to ensure congruence. Instead, one can add some compensatory seats and then design the incongruence breaking algorithm such that proportionality has priority over regional or geographic representation (which means that some one-tier seats might be taken back if the national and regional results are incongruent). 

To our knowledge, guaranteed proportionality has so far only been employed in Sweden and Germany building on multi- and single-member constituencies respectively. 
For guaranteed proportional systems it is also important to determine the number of compensatory seats, since all regional seats can be honoured as long as the results are congruent. A potential future analysis could show how many constituency winners would not make it to parliament (due to lacking \emph{Zweitstimmendeckung} in the German nomenclature). 

Designers of electoral systems aiming for increased proportionality and regionality should not necessarily confine themselves to the setup of classical two-tier electoral system. 
A common theme of alternative frameworks is that the dichotomy between regional seats and compensatory seats is blurred. 

The most used alternative proposal is \emph{double proportionality}, which is used by several Swiss cantons, see e.g., \cite[Chap. 14-15]{pukelsheim2017proportional} for a detailed mathematical description.
But many other systems have been proposed \cite{holdum2025impossibility,linusson2014dynamic}.

\subsection{Outlook} 
The formulas for the accumulated seat surplus in one-tier electoral systems are more precise when multi-member constituencies are larger. 
For this reason, the important case of single member constituencies (FPTP) was excluded from the analysis. Obtaining any sensible handle on this case is an interesting avenue for further work and would directly inform the work on electoral reform in Germany and complement the work in \cite{bochsler2023balancing}. 

There is also room for substantial further development of the mathematical work. For simplicity, many formulas were only given for Hare Quota, but generalize to the entire families of $\beta$-divisor and $\gamma$-quota methods. 
In addition, the distributional convergence of \cite{janson2014asymptotic} can be used to estimate the expected variance of $L^*$ and help deal with the 'statistical fluctuations' part of this paper. 

This analysis would also open up for estimating the probability of incongruence. 
This probability was assessed for the Danish electoral system in \cite{klausensandsynligheden} using a slightly more empirical method, which built on a numerical study of the deviations from \eqref{eq:expected_bias}. This method can be generalized to the category of two-tier electoral systems studied in this paper.

Context matters.  The apportionment of the constituency seats, regionally strong parties and island constituencies with guaranteed seats are important for the legitimacy of specific systems. The resulting malapportionment transfer is an important factor for determining the number of compensatory seats in an electoral system.  

For each country currently employing an electoral system with multi-member constituencies the generic analysis made here can be sharpened in that specific context.

\section*{AI statement}
This paper has been powered by AI in many aspects. In particular, the figures and tables were AI-generated, but human-checked, human-curated and human-adjusted. 
The text and formulas are human-written with AI feedback.  
The AI-generated code generating the figures is available on GitHub \href{https://github.com/FrederikRavnKlausen/Ensuring-proportionality-a-logical-model-for-compensatory-seats-added-to-multi-member-constituencies}{here}. 

\section*{Acknowledgments}
The author would like to thank Jørgen Elklit, Sebastian Holdum and Theodora Helimäki for helpful discussions.
This work would have been impossible without the extensive database \cite{kollman2024clea}. Claude Code has sped up the data analysis and figure generation dramatically. 
The author was supported by the Carlsberg Foundation, CF24-0466.

\bibliographystyle{apalike}
\bibliography{bibliography}

\end{document}